\documentclass[preprint,11pt]{elsarticle}
\usepackage[letterpaper,margin=1in]{geometry}

\usepackage{amsmath,amsfonts,amssymb}
\usepackage{graphicx}
\usepackage{bm}
\usepackage{booktabs}
\usepackage{longtable,tabularx}
\usepackage{amsthm}
\newtheorem{theorem}{Theorem}
\newtheorem{lemma}{Lemma}
\newtheorem{corollary}{Corollary}
\newtheorem{proposition}{Proposition}
\theoremstyle{remark}
\newtheorem{remark}{Remark}

\newcommand{\Exp}{\operatorname{Exp}}

\newcommand{\ad}{\operatorname{ad}}

\newcommand{\SE}{\mathrm{SE}_2(3)}
\newcommand{\SO}{\mathrm{SO}(3)}
\newcommand{\skewmat}[1]{[#1]_\times}

\journal{Acta Astronautica}

\begin{document}

\begin{frontmatter}

\title{Two-impulse rendezvous planning about thrusting spacecraft on $\mathrm{SE}_2(3)$}

\author[pu]{Micah Condie\corref{cor1}}
\ead{condiem@purdue.edu}
\author[pu]{Abigaile Woodbury}
\author[pu]{James Goppert}
\cortext[cor1]{Corresponding author.}
\affiliation[pu]{organization={School of Aeronautics and Astronautics, Purdue University},
            city={West Lafayette},
            state={IN},
            postcode={47907},
            country={USA}}

\begin{abstract}
The classical Hill--Clohessy--Wiltshire equations assume an unforced Keplerian reference trajectory, an assumption that is violated by missions requiring continuous thrust. We formulate spacecraft relative motion on the $\mathrm{SE}_2(3)$ Lie group, which encodes position, velocity, and attitude in a single geometric state, linearize the gravity and body-frame control mismatches between the two vehicles, and derive analytic bounds on every neglected term. Under circular coasting Keplerian assumptions the framework recovers the Hill--Clohessy--Wiltshire equations exactly. For thrusting references the state transition matrix acquires thrust-induced attitude--translation coupling blocks that, when the two vehicles share body-frame inputs, are exact in log coordinates and independent of the attitude error, and a mismatch between the vehicles' inputs re-centers the linearization at their mean. A two-impulse rendezvous planner follows directly. A scalar coupling ratio, computable before planning, separates two regimes. Where it exceeds unity, the planner reduces the terminal error of coupling-blind planners by more than an order of magnitude, and somewhat above unity it is also more accurate than a Cartesian planner of the same computational structure; well below unity, simpler planners are adequate or better. Simulations confirm the recovery to machine precision and validate every bound along the transfer.
\end{abstract}


\end{frontmatter}

\section*{Nomenclature}

{\renewcommand\arraystretch{1.0}
\noindent\begin{longtable}{@{}l @{\quad=\quad} p{0.68\textwidth}@{}}
$a,\,\bar{a}$        & actual and reference body-frame accelerations, m/s$^2$ \\
$a_{\mathrm m}$      & mean of actual and reference body-frame accelerations, m/s$^2$ \\
$A(t),\,\tilde{A}(t)$ & system matrices of the matched- and mismatched-control systems, s$^{-1}$ \\
$A_C$                & kinematic coupling matrix encoding $\dot{\xi}_p = \xi_v$ \\
$\mathrm{ad}_\xi$    & Lie algebra adjoint operator \\
$\mathrm{Ad}^\vee_X$ & group adjoint map on coordinate vectors \\
$b(T,t_0)$           & particular integral of the forced log-error system \\
$c(\theta)$          & second-order coefficient of the inverse left Jacobian on $\mathrm{SO}(3)$, rad$^{-2}$ \\
$c_{\mathbb C}$      & complex representation of the left Jacobian on the plane $\hat\xi^\perp$ \\
$\mathcal{C}(T)$     & total impulse cost $\|\delta v_0\| + \|\delta v_T\|$, m/s \\
$d$                  & relative separation $\|\rho\|$, m \\
$\mathrm{Exp},\,\mathrm{Log}$ & group exponential and logarithmic maps \\
$F(t)$               & attitude sensitivity of the linearized Cartesian thrust forcing, m/(s$^2\cdot$rad) \\
$G(t)$               & body-frame tidal tensor, s$^{-2}$ \\
$G_I(t)$             & inertial-frame tidal (gravity-gradient) tensor, s$^{-2}$ \\
$h(\theta)$          & auxiliary function establishing $m_a > 2m_p$ \\
$J_\ell,\,J_r$       & left and right Jacobians of $\mathrm{SE}_2(3)$ \\
$k_b$                & forcing contribution to the arrival velocity error, m/s \\
$K_p,\,K_R$          & arrival velocity gains on initial position and attitude error, s$^{-1}$ and m/(s$\cdot$rad) \\
$L(\tilde{\nu})$     & linear-in-$\xi$ control mismatch correction matrix, s$^{-1}$ \\
$m,\,\bar{m}$        & actual and reference inertial-frame gravity vectors in $\mathbb{R}^9$ \\
$m_a,\,m_p$          & axial and planar coefficients of $\mathcal{M}$, rad$^{-1}$ \\
$\mathcal{M}(u,\xi_R)$ & deviation of the $Q$-route from its zero-attitude limit \\
$n$                  & orbital mean motion, rad/s \\
$p,\,\bar{p}$        & actual and reference inertial-frame position vectors, m \\
$Q(u,\xi_R)$         & off-diagonal block generator of the $\mathrm{SE}_2(3)$ left Jacobian \\
$r$                  & reference orbital radius $\|\bar{p}\|$, m \\
$\hat{r}$            & body-frame radial unit vector \\
$R,\,\bar{R}$        & actual and reference body-to-inertial attitude matrices, $\in\mathrm{SO}(3)$ \\
$S$                  & skew matrix $[\xi_R]_\times$ of the attitude error, rad \\
$\mathrm{SE}_2(3)$   & group of double direct spatial isometries \\
$\mathfrak{se}_2(3)$ & Lie algebra of $\mathrm{SE}_2(3)$ \\
$\mathrm{SO}(3)$     & rotation group \\
$T$                  & terminal time of the transfer horizon, s \\
$T_c$                & similarity transform relating log and classical coordinates \\
$u$                  & generic translation argument of $Q$ and $\mathcal{M}$ \\
$v,\,\bar{v}$        & actual and reference inertial-frame velocity vectors, m/s \\
$X,\,\bar{X}$        & actual and reference spacecraft states in $\mathrm{SE}_2(3)$ \\
$[\,\cdot\,]_\times$ & skew-symmetric matrix, $[x]_\times y = x\times y$ \\*[4pt]
$\alpha(\theta),\,\beta(\theta)$ & control mismatch residual bound coefficients, rad$^{-1}$ and dimensionless \\
$\gamma_1,\,\gamma_2$ & coefficients of the $\mathrm{SO}(3)$ left Jacobian, rad$^{-2}$ and rad$^{-3}$ \\
$\delta_{\mathrm{att}}$ & attitude-induced gravity linearization error, m/s$^2$ \\
$\delta_{\mathrm{grav}}$ & higher-order gravity linearization error, m/s$^2$ \\
$\delta v$           & physical impulsive velocity increment, m/s \\
$\Delta a,\,\Delta\omega$ & acceleration and angular velocity mismatch, actual minus reference \\
$\Delta t$           & transfer duration $T - t_0$, s \\
$\Delta\xi_v$        & impulsive correction in log-velocity coordinates, m/s \\
$\varepsilon_p,\,\varepsilon_v,\,\varepsilon_R$ & control mismatch linearization residuals \\
$\eta$               & left-invariant configuration error in $\mathrm{SE}_2(3)$ \\
$\theta,\,\varphi$   & attitude-error magnitude $\|\xi_R\|$, rad, and its angle from $\hat r$, rad \\
$\kappa$             & coupling-to-tidal acceleration ratio \\
$\mu$                & Earth gravitational parameter, m$^3$/s$^2$ \\
$\nu,\,\bar{\nu}$    & actual and reference body-frame input vectors in $\mathbb{R}^9$ \\
$\nu_{\mathrm m}$    & mean of actual and reference body-frame input vectors \\
$\tilde{\nu},\,\tilde{m}$ & body-frame input and inertial-frame gravity mismatch vectors \\
$\xi=(\xi_p,\xi_v,\xi_R)$ & log-error coordinates in $\mathbb{R}^9$, (m, m/s, rad) \\
$\rho$               & relative position vector $p-\bar p$, m \\
$\rho_{\mathrm{target}}$ & commanded standoff position, m \\
$\Phi(t,t_0),\,\tilde{\Phi}(t,t_0)$ & state transition matrices of the free and forced systems \\
$\Psi$               & attitude-induced commutator term in $\delta_{\mathrm{att}}$, s$^{-2}$ \\
$\Omega$             & local-vertical local-horizontal angular velocity vector, rad/s \\
$\omega,\,\bar{\omega}$ & actual and reference body-frame angular velocities, rad/s \\
$\omega_{\mathrm m}$ & mean of actual and reference body-frame angular velocities, rad/s \\
\end{longtable}}

\setcounter{table}{0}

\section{Introduction}
\label{sec:intro}

Proximity operations, considered to be any space operation where two or more satellites are within 500 km~\cite{petersen2024}, include rendezvous missions, formation flying, inspection, and active debris removal. These operations underpin a growing class of space missions whose complexity continues to increase~\cite{woffinden2007,petersen2024}.
Accurate relative motion models are essential for efficient maneuver planning in these scenarios, where small errors in predicted trajectories translate directly into propellant penalties or mission failure.
The Hill--Clohessy--Wiltshire (HCW) equations~\cite{clohessy1960} remain the workhorse of proximity operations planning due to their analytical tractability: they admit closed-form state transition matrices and explicit two-impulse solutions that are fast enough for onboard implementation.
Their accuracy, however, rests on the critical assumption that the spacecraft follow an unforced Keplerian trajectory.
This assumption is increasingly violated in modern missions, including continuous station-keeping in geosynchronous orbit and low-thrust electric-propulsion transfers~\cite{garulli2011}.
 
Extensions to HCW have addressed many of its other limitations while retaining the unforced spacecraft assumption.
Curvilinear generalizations improve accuracy at large separations; the Schweighart--Sedwick equations incorporate $J_2$ perturbations~\cite{schweighart2002}; and the Yamanaka--Ankersen state transition matrix~\cite{yamanaka2002} extends HCW to eccentric reference orbits.
Relative orbital element (ROE) parameterizations~\cite{dAmico2010,gaias2015} provide geometrically meaningful state representations that have been validated in flight and extended to handle $J_2$-perturbed and eccentric orbits.
The comprehensive survey of Sullivan, Grimberg, and D'Amico~\cite{sullivan2017} assesses the full landscape of these models and their domains of validity.
Yet all of these formulations share the foundational requirement that the trajectory is unaccelerated, coasting under gravity alone.
When one, or both, of the spacecraft considered in a docking problem is thrusting, none of these models correctly captures the dynamics, and the resulting errors can be severe, as demonstrated in Section~\ref{sec:results} for the HCW case.
 
A parallel line of work has approached spacecraft relative motion through geometric and Lie group methods, motivated by the coordinate singularities and linearization inconsistencies that arise in Euclidean formulations.
Lee, Leok, and McClamroch~\cite{lee2007} developed Lie group variational integrators for the full two-body orbital mechanics problem on $\mathrm{SE}(3)$, providing geometrically exact simulations that preserve symplecticity and conserved quantities over long time spans without reprojection.
Filipe and Tsiotras~\cite{filipe2015} formulated coupled six-degree-of-freedom proximity operations dynamics using dual quaternions, unifying rotational and translational motion into a single kinematic model and demonstrating adaptive tracking controllers for satellite servicing. The invariant observer literature~\cite{barrau2015,barrau2017} introduced the notion of group-affine systems on Lie groups and showed that the resulting error dynamics admit a log-linear structure. Recent work on multirotor systems~\cite{lin2023} demonstrated that dynamics embedded on $\mathrm{SE}_2(3)$ are exactly group-affine when gravity is treated as
a constant, yielding exact log-linear error dynamics.
Building on this,~\cite{condie_loglinear} established that spacecraft orbital dynamics admit the same $\mathrm{SE}_2(3)$ formulation, with the position-dependent gravity term acting as a barrier to the system being group-affine. The resulting log-error dynamics were shown to be linear in the error state, with the exception of a nonlinear gravity mismatch term that is analytically bounded. The present paper is a direct sequel to~\cite{condie_loglinear} that linearizes and bounds the higher order terms of the error dynamics. It establishes its significance to the spacecraft docking problem by using the framework to develop a two-impulse rendezvous planner.

Classical planners such as HCW are fast and simple, but they do not represent body-frame acceleration or the attitude--translation coupling it induces. That coupling can also be carried in Cartesian coordinates, as a thrust forcing that depends on the relative attitude. The question this paper addresses is therefore not only whether the coupling matters, but which coordinates represent it best. We provide a framework that accommodates thrusting spacecraft with bounded linearization error, recovers HCW exactly under an unforced Keplerian reference, and, when the two vehicles share body-frame inputs, carries the coupling exactly, as a fixed block of a state transition matrix that does not depend on the attitude error. A scalar criterion, computable before planning, identifies when this representation is preferable to a Cartesian one and when a classical planner suffices.

Concretely, the contributions are:
\begin{enumerate}
\item \textit{Gravity linearization with analytic error bounds} (Theorem~\ref{thm:gravity_linearization}). We linearize the gravity mismatch within the $\mathrm{SE}_2(3)$ log-error framework, decomposing the error into an attitude-induced component and a higher-order separation-dependent component. Analytic bounds are derived for each, tight in the configurations that attain them, and numerical experiments confirm both are small fractions of the dominant tidal term at typical proximity-operations separations.

\item \textit{Exact HCW recovery as a special case} (Theorem~\ref{thm:hcw}). Under circular coasting Keplerian assumptions, the linear time varying system reduces exactly to the classical HCW equations when converted to the LVLH frame. This establishes HCW as a special case of the $\mathrm{SE}_2(3)$ framework rather than an independent linearization.

\item \textit{Control mismatch linearization with analytic error bounds} (Theorem~\ref{thm:control_mismatch}). We linearize the body-frame control mismatch term---both the acceleration and the angular velocity components---in the $\mathrm{SE}_2(3)$ log-error framework to first order, much like the gravity term, and bound the residual in each block. Absorbing the linear part re-centers the system matrix at the mean of the two vehicles' inputs. This allows for computationally tractable solutions to docking problems where the two spacecraft have different body-frame control.

\item \textit{Two-impulse rendezvous planning for thrusting references}. A two-impulse planner derived directly from the state transition matrix (STM) that follows from the $\mathrm{SE}_2(3)$ log-error framework generalizes naturally to thrusting reference trajectories by accounting for attitude--translation coupling through off-diagonal STM blocks that are absent in all classical formulations. The planner is derived for cases with and without body-frame control mismatch.

\item \textit{A regime criterion and a fair comparison}. The coupling ratio $\kappa$ of Remark~\ref{rem:coupling} classifies each scenario before planning. Simulations show that where $\kappa \gtrsim 1$ the planner reduces the terminal error of coupling-blind planners by factors of $17$ to $181$, and that above a crossover between $\kappa = 1.4$ and $1.8$ it is also more accurate than a Cartesian planner of the same computational structure; below that crossover the Cartesian planner, and at $\kappa \ll 1$ in some geometries a classical one, is preferable. Remark~\ref{rem:cartesian} explains why: each formulation incurs a different second-order residual, and the ratio of the two is of order $\kappa$.
\end{enumerate}
 
The paper is organized as follows.
Section~\ref{sec:prelim} establishes notation and derives the log-error dynamics of~\cite{condie_loglinear}.
Section~\ref{sec:gravity_linearization} produces the gravity linearization and proves the error bounds of Theorem~\ref{thm:gravity_linearization}.
Section~\ref{sec:hcw_recovery} shows recovery of HCW as a special case (Theorem~\ref{thm:hcw}). Section~\ref{sec:control_mismatch} introduces the control mismatch linearization and the bounds on its higher order terms (Theorem~\ref{thm:control_mismatch}). Section~\ref{sec:dv_planning} develops the two-impulse planning method based on the log-error dynamics. Section~\ref{sec:results} presents numerical validation, and Section~\ref{sec:conclusions} concludes.


\section{Preliminaries}
\label{sec:prelim}

\subsection{Lie-Group Preliminaries}
\label{sec:lie-prelim}

The configuration of a rigid body, such as a spacecraft, does not live in a
vector space. While velocity and position can be added and subtracted
component-wise, attitude cannot. Lie-group theory provides the geometric
framework to handle this correctly. A Lie group $G$ is simultaneously a smooth
manifold and a group; for the matrix groups used in this work, elements compose
by matrix multiplication and invert by matrix inversion. The group product
supplies a consistent notion of difference in place of subtraction.

The rotation group $\mathrm{SO}(3)$ consists of the $3 \times 3$ orthogonal
matrices with determinant one. When representing rigid-body attitude, these are
commonly known as direction cosine matrices. The group $\mathrm{SE}_2(3)$ of
double direct spatial isometries~\cite{barrau2015} augments a rotation with two
translational components, allowing a single element
\begin{equation}
  X =
  \begin{bmatrix}
    R & v & p \\
    0_{1\times3} & 1 & 0 \\
    0_{1\times3} & 0 & 1
  \end{bmatrix}
  \in \mathrm{SE}_2(3),
  \label{eq:Xblock}
\end{equation}
to capture attitude $R \in \mathrm{SO}(3)$, velocity $v \in \mathbb{R}^3$, and
position $p \in \mathbb{R}^3$ together. We adopt $\mathrm{SE}_2(3)$ as the
configuration space for the spacecraft state.

Associated with every Lie group is its Lie algebra $\mathfrak{g}$, the tangent
space at the identity. Unlike the group, the algebra is a vector space. The
exponential and logarithm,
\begin{equation}
  \mathrm{Exp} : \mathfrak{g} \to G, \qquad \mathrm{Log} : G \to \mathfrak{g},
  \label{eq:explog}
\end{equation}
connect the two; for the matrix groups used here, $\mathrm{Exp}$ is the matrix
exponential. A configuration error is thus formed in the group and mapped by
$\mathrm{Log}$ into the algebra, where it becomes an ordinary vector in
$\mathbb{R}^9$ suitable for analysis. The algebra associated with the group $\mathrm{SE}_2(3)$ is $\mathfrak{se}_2(3)$.

The algebra carries additional structure used throughout. The hat map
$(\cdot)^\wedge : \mathbb{R}^9 \to \mathfrak{se}_2(3)$ sends a coordinate vector
$\xi = (\xi_p, \xi_v, \xi_R) \in \mathbb{R}^9$, with position, velocity, and
attitude blocks, to the algebra element
\begin{equation}
  \xi^\wedge =
  \begin{bmatrix}
    [\xi_R]_\times & \xi_v & \xi_p \\
    0_{1\times3} & 0 & 0 \\
    0_{1\times3} & 0 & 0
  \end{bmatrix}
  \in \mathfrak{se}_2(3),
  \label{eq:xiblock}
\end{equation}
with inverse the vee map $(\cdot)^\vee$; here $[\,\cdot\,]_\times$ is the
skew-symmetric operator mapping a vector to the matrix representing its cross
product. The algebra is closed under the matrix commutator $[A,B] = AB - BA$, where $A,B \in \mathfrak{se}_2(3)$.
This defines the algebra adjoint $\mathrm{ad}_A B = [A,B]$. Its group-level
counterpart, the group adjoint $\mathrm{Ad}_X B = X B X^{-1}$, transports
algebra elements between frames.
This is a linear operator, acting on algebra coordinate vectors through a $9 \times 9$ matrix (for the $\mathrm{SE}_2(3)$ group) $\mathrm{Ad}_X^\vee$.

\subsection{Reference Frames and Spacecraft State}
\label{sec:frames}

Three right-handed Cartesian frames are used throughout, illustrated in
Fig.~\ref{fig:frames}. The inertial frame $\mathcal{I}$ is an Earth-Centered
Inertial (ECI) frame with origin at the Earth's center and components $\hat e_1, \hat e_2$, and $\hat e_3$. Each spacecraft carries
a right-handed body frame $\mathcal{B} = \{\hat b_1, \hat b_2, \hat b_3\}$ and $\mathcal{\bar B} = \{\hat{\bar{b_1}}, \hat{\bar{b_2}}, \hat{\bar{b_3}}\}$ fixed to the vehicle, with origin at the
center of mass. The control acceleration $a$ and angular velocity $\omega$ are
expressed in the body frame.

The spacecraft configuration is the group element $X \in \mathrm{SE}_2(3)$ of
Eq.~\eqref{eq:Xblock}, with body-to-inertial attitude $R \in \mathrm{SO}(3)$,
inertial-frame velocity $v \in \mathbb{R}^3$, and inertial-frame position
$p \in \mathbb{R}^3$. A second trajectory $\bar X(t) \in \mathrm{SE}_2(3)$, with
components $(\bar R, \bar v, \bar p)$ in the same block form, is the reference trajectory; $X(t)$ is the deputy. For the purpose of the docking problem, we sometimes refer to $\bar X$ as the chief and $X$ as the deputy.
%

\begin{figure}
\centering
\includegraphics[width=0.85\columnwidth]{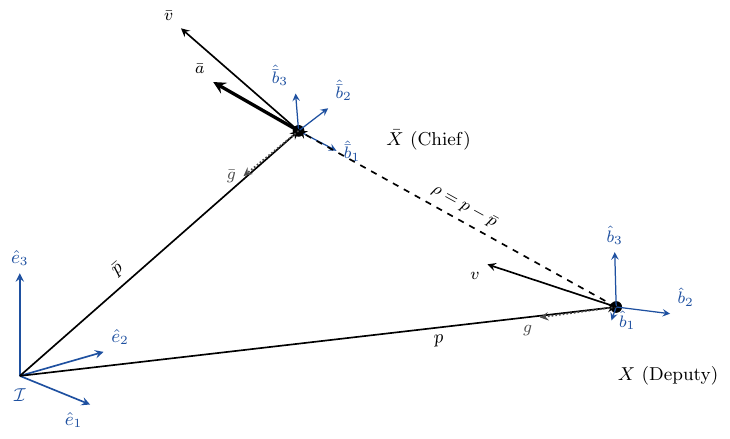}

\caption{Spacecraft state on $\mathrm{SE}_2(3)$ and chief--deputy geometry.
Coordinate frames are shown in blue. The inertial frame $\mathcal{I}$ is
Earth-Centered Inertial, with Earth at the origin. The chief (reference) state
$\bar{X}$ has body frame $\bar{\mathcal{B}}$ with axes $\{\hat{\bar{b}}_1,
\hat{\bar{b}}_2, \hat{\bar{b}}_3\}$, inertial position $\bar{p}$, inertial
velocity $\bar{v}$, body-frame acceleration $\bar{a}$, and experiences
gravitational acceleration $\bar{g} = g(\bar{p})$ directed toward Earth. The
deputy state $X$ has analogous quantities $\mathcal{B}$, $\{\hat{b}_1,
\hat{b}_2, \hat{b}_3\}$, $p$, $v$, and $g = g(p)$. The body frames are
generally distinct: the attitudes $\bar{R}, R \in \mathrm{SO}(3)$ rotate
vectors from $\bar{\mathcal{B}}$ and $\mathcal{B}$ respectively into
$\mathcal{I}$, so the two triads are oriented differently with respect to
$\mathcal{I}$. The relative position $\rho = p-\bar{p}$ enters the gravity
Taylor expansion of Section~\ref{sec:gravity_linearization}. The invariant
configuration error $\eta = \bar{X}^{-1}X$ is expressed in the chief's body
frame.}
\label{fig:frames}
\end{figure}

\subsection{Inputs and Equations of Motion}
\label{sec:eom}

Both spacecraft are driven by body-frame inputs. Let $a \in \mathbb{R}^3$
denote the body-frame acceleration and $\omega \in \mathbb{R}^3$ the
body-frame angular velocity of the deputy, with $\bar a$, $\bar\omega$ the
corresponding quantities for the chief. The equations of motion are
\begin{equation}
  \dot p \;=\; v, \qquad
  \dot v \;=\; R\, a \;+\; g(p), \qquad
  \dot R \;=\; R\, [\omega]_\times,
  \label{eq:eom-actual}
\end{equation}
with the analogous equations holding for $(\bar p, \bar v, \bar R)$ driven
by $(\bar a, \bar\omega)$. The matrix $R$ rotates the body-frame thrust $a$
into the inertial frame, and $[\,\cdot\,]_\times$ is the skew-symmetric map of
Sec.~\ref{sec:lie-prelim}. The gravitational acceleration
$g(p) = -\mu\, p / \|p\|^3$, with $\mu$ the Earth's gravitational parameter,
is defined in the inertial frame and directed from the spacecraft toward the
center of the Earth (Fig.~\ref{fig:frames}).

Following~\cite{condie_loglinear}, Eqs.~\eqref{eq:eom-actual} admit the
compact group form
\begin{equation}
  \dot X \;=\; (M-C)X \;+\; X(N+C),
  \label{eq:Xdot}
\end{equation}
where
\begin{equation}
  M \;=\;
  \begin{bmatrix}
    0_{3\times 3} & g(p) & 0_{3\times 1}\\
    0_{1\times 3} & 0 & 0\\
    0_{1\times 3} & 0 & 0
  \end{bmatrix},\quad
  N \;=\;
  \begin{bmatrix}
    [\omega]_\times & a & 0_{3\times 1}\\
    0_{1\times 3} & 0 & 0\\
    0_{1\times 3} & 0 & 0
  \end{bmatrix},\quad
  C \;=\;
  \begin{bmatrix}
    0_{3\times 3} & 0_{3\times 1} & 0_{3\times 1}\\
    0_{1\times 3} & 0 & 1\\
    0_{1\times 3} & 0 & 0
  \end{bmatrix},
  \label{eq:MNC}
\end{equation}
with $\bar M, \bar N$ defined analogously. Here $M$ carries the
inertial-frame gravitational acceleration and $N$ carries the body-frame
inputs; the constant matrix $C$ encodes the kinematic coupling
$\dot p = v$ but is not an element of $\mathfrak{se}_2(3)$.

The body-frame input vector and its reference are
\begin{equation}
  \nu = N^\vee \;=\;
  \begin{bmatrix} 0 \\ a \\ \omega \end{bmatrix},
  \qquad
  \bar \nu = \bar N^\vee \;=\;
  \begin{bmatrix} 0 \\ \bar a \\ \bar\omega \end{bmatrix},
  \label{eq:ndef}
\end{equation}
where $(\cdot)^\vee$ is the vee map of Sec.~\ref{sec:lie-prelim}. The leading
three-zero block reflects that position is not directly forced by a
body-frame input; the kinematic coupling $\dot p = v$ is carried by $C$.
\subsection{Configuration-Tracking Error}
\label{sec:error}

The configuration error compares the actual state to the reference state.
On $\mathrm{SE}_2(3)$, two natural choices are available~\cite{barrau2017}:
the left-invariant error $X^{-1}\bar X$ and its inverse $\bar X^{-1} X$. We
adopt
\begin{equation}
  \eta = \bar X^{-1} X \;\in\; \mathrm{SE}_2(3),
  \label{eq:etadef}
\end{equation}
which expresses the deputy state in the chief's frame. The position
component of the associated log-coordinates then represents the
deputy-minus-chief position resolved in the chief's frame, matching the
classical Hill--Clohessy--Wiltshire sign convention.\footnote{The prior
work~\cite{condie_loglinear} used the opposite convention
$\eta = X^{-1}\bar X$. The dynamics derived below in
Proposition~\ref{prop:loglinear} can be obtained from the corresponding
result of~\cite{condie_loglinear} by swapping actual and reference
quantities.}

The log-error state is
\begin{equation}
  \xi = \mathrm{Log}(\eta)^\vee
  \;=\;
  \begin{bmatrix} \xi_p \\ \xi_v \\ \xi_R \end{bmatrix}
  \;\in\; \mathbb{R}^9,
  \label{eq:xidef}
\end{equation}
where $\xi_p, \xi_v, \xi_R \in \mathbb{R}^3$ are the position, velocity,
and attitude error coordinates. The error $\eta$ is invariant under
simultaneous left-translation of both states by any group element, so
$\xi$ depends only on the relative configuration and not on the absolute
inertial pose. The representation~\eqref{eq:xidef} is well-defined for
$\|\xi_R\| < \pi$, which is benign for proximity operations.

\subsection{Log-Error Dynamics}
\label{sec:logerror}

\begin{proposition}[Log-error dynamics on $\mathrm{SE}_2(3)$]
\label{prop:loglinear}
Let $X(t)$ and $\bar X(t)$ evolve under~\eqref{eq:eom-actual} and its
overlined counterpart, and define the input and gravity mismatch vectors
\begin{equation}
  \tilde \nu = \nu - \bar \nu =
  \begin{bmatrix} 0 \\ a - \bar a \\ \omega - \bar \omega \end{bmatrix},
  \qquad
  \tilde m = m - \bar m =
  \begin{bmatrix} 0 \\ g(p) - g(\bar p) \\ 0 \end{bmatrix}.
  \label{eq:tildes}
\end{equation}
Then the log-error~\eqref{eq:xidef} satisfies
\begin{equation}
  \dot\xi
  \;=\;
  \bigl( -\mathrm{ad}_{\nu} + A_C \bigr)\, \xi
  + J_\ell^{-1}(\xi)\,\tilde \nu
  + J_r^{-1}(\xi)\,\mathrm{Ad}^\vee_{X^{-1}}\,\tilde m,
  \label{eq:loglinear}
\end{equation}
where $\mathrm{ad}_{\nu} : \mathbb{R}^9 \to \mathbb{R}^9$ is the matrix
representation of the Lie-algebra adjoint, $\mathrm{Ad}^\vee_{X^{-1}}$ is
the group-adjoint matrix induced by $X^{-1}$, $J_\ell(\xi)$ and $J_r(\xi)$
are the left and right Jacobians of $\mathrm{SE}_2(3)$ at $\xi$, and
\begin{equation}
  A_C
  \;=\;
  \begin{bmatrix}
    0_{3\times 3} & I_3 & 0_{3\times 3} \\
    0_{3\times 3} & 0_{3\times 3} & 0_{3\times 3} \\
    0_{3\times 3} & 0_{3\times 3} & 0_{3\times 3}
  \end{bmatrix}
  \label{eq:AC}
\end{equation}
encodes the kinematic coupling $\dot\xi_p = \xi_v$ in the log-coordinates.
\end{proposition}

\begin{proof}
Differentiating $\eta = \bar X^{-1} X$ and using
$\tfrac{d}{dt}(\bar X^{-1}) = -\bar X^{-1}\dot{\bar X}\,\bar X^{-1}$
together with~\eqref{eq:Xdot} gives
\begin{equation*}
  \dot\eta
  \;=\; \bar X^{-1}(M-\bar M)\bar X\, \eta
  \;+\; \eta\, N \;-\; \bar N\, \eta
  \;+\; \eta\, C \;-\; C\, \eta.
\end{equation*}
Defining $\tilde M =  M - \bar M$ and $\tilde N = N - \bar N$
and left-multiplying by $\eta^{-1}$ yields
\begin{equation*}
  \eta^{-1}\dot\eta
  \;=\; \mathrm{Ad}_{X^{-1}}\tilde M
  \;+\; (I - \mathrm{Ad}_{\eta^{-1}})\, N
  \;+\; \mathrm{Ad}_{\eta^{-1}}\, \tilde N
  \;+\; (C - \eta^{-1} C \eta),
\end{equation*}
where the identity $\eta^{-1}\bar X^{-1} = X^{-1}$ has been used to
collapse the gravity adjoint. As shown
in~\cite[App.~A]{condie_loglinear},
$C - \eta^{-1} C \eta = (I - e^{-\mathrm{ad}_{\xi^\wedge}})\, C \in \mathfrak{se}_2(3)$,
so the vee operator may be applied. The right-trivialized differential of
the log map gives $\dot\xi^\wedge = J_r^{-1}(\xi^\wedge)\bigl(\eta^{-1}\dot\eta\bigr)$;
combining with $\mathrm{Ad}_{\eta^{-1}} = e^{-\mathrm{ad}_{\xi^\wedge}}$
and $J_r^{-1}(\xi^\wedge)\bigl(I - e^{-\mathrm{ad}_{\xi^\wedge}}\bigr) = \mathrm{ad}_{\xi^\wedge}$
reduces the $N$ contribution to $\mathrm{ad}_{\xi^\wedge} N = -\mathrm{ad}_N \xi^\wedge$.
The input mismatch contribution is transported by the left/right
identity~\cite[\S7.1.5]{barfoot2024}
\begin{equation*}
  J_r^{-1}(\xi)\,\mathrm{Ad}^\vee_{\eta^{-1}} \;=\; J_\ell^{-1}(\xi),
\end{equation*}
which converts $J_r^{-1}(\xi)\,\mathrm{Ad}^\vee_{\eta^{-1}}\tilde\nu$ into
$J_\ell^{-1}(\xi)\,\tilde\nu$. Applying the vee operator and using
$(\mathrm{ad}_{\xi^\wedge} C)^\vee = A_C \xi$
from~\cite[App.~A]{condie_loglinear} yields~\eqref{eq:loglinear}.
\end{proof}

The matrix $-\mathrm{ad}_{\nu} + A_C$ generates the body-frame rotation
cross-products and the kinematic coupling $\dot\xi_p = \xi_v$. The term
$J_\ell^{-1}(\xi)\,\tilde \nu$ captures the body-frame input mismatch
between deputy and chief, and is the focus of
Section~\ref{sec:control_mismatch}. The term
$J_r^{-1}(\xi)\,\mathrm{Ad}^\vee_{X^{-1}}\,\tilde m$
is the inertial-frame gravity mismatch, the sole obstruction to the
dynamics being exactly group-affine in the sense
of~\cite{barrau2017}, and the focus of
Section~\ref{sec:gravity_linearization}.

\subsection{Action of the Gravity Term on Velocity Coordinates}
\label{sec:gravblock}

\begin{lemma}[Velocity-block structure of the gravity term]
\label{lem:gravblock}
The right Jacobian $J_r^{-1}(\xi)$ on $\mathrm{SE}_2(3)$ is block
upper-triangular in the $(\xi_p,\xi_v,\xi_R)$ ordering, with identical
diagonal blocks $J_r^{-1,\mathrm{SO}(3)}(\xi_R)$. Consequently,

\begin{equation}
\label{eq:gravity_block}
  J_r^{-1}(\xi)\,\mathrm{Ad}^\vee_{X^{-1}}\,\tilde m
  \;=\;
  \begin{bmatrix}
    0 \\[2pt]
    J_r^{-1,\mathrm{SO}(3)}(\xi_R)\, R^\top\bigl(g(p) - g(\bar p)\bigr) \\[2pt]
    0
  \end{bmatrix}.
\end{equation}
\end{lemma}

\begin{proof}
The matrix $\mathrm{ad}_\xi$ on $\mathfrak{se}_2(3)$ is block
upper-triangular with diagonal blocks $[\xi_R]_\times$, so all of its
powers, and hence the Bernoulli series $J_r^{-1}(\xi) = \sum_{k=0}^\infty (B_k^+/k!)\,\mathrm{ad}_\xi^k$,
share this structure with diagonal blocks
$J_r^{-1,\mathrm{SO}(3)}(\xi_R)$~\cite[Prop.~1]{condie_loglinear}. The
vector $\mathrm{Ad}^\vee_{X^{-1}}\,\tilde m$ has only a velocity entry,
equal to $R^\top\bigl(g(p) - g(\bar p)\bigr)$, and block upper-triangular
multiplication extracts only the diagonal block.
\end{proof}


\section{Gravity Linearization and Error Bounds}
\label{sec:gravity_linearization}

This section linearizes the gravity mismatch in~\eqref{eq:gravity_block},
derives bounds on the linearization error, tight in the configurations that attain
them, and establishes the linear
time-varying system used for planning. Throughout, $\theta = \|\xi_R\|$,
$S = [\xi_R]_\times$, $\hat\xi = \xi_R/\theta$, $r = \|\bar{p}\|$,
$\hat{r} := \bar{R}^\top\bar{p}/r$ denotes the body-frame radial unit vector,
and $G(t) := \bar{R}^\top G_I(t)\bar{R}$ is the body-frame gravity-gradient
tensor (with $G_I$ defined in~\eqref{eq:gravity_taylor}). We
assume $\theta \in (0,\pi)$, consistent with the injectivity radius of the log
map used elsewhere. The exclusion of $\theta = 0$ is made only so that
$\hat\xi$ and the quotients in $\theta$ appearing below are defined. Every
quantity in this section and in Section~\ref{sec:control_mismatch} extends
continuously to $\theta = 0$, where $J_\ell^{\mathrm{SO}(3)} = I$, $\Psi = 0$,
and the bound coefficients $\alpha,\beta$ of~\eqref{eq:alpha_beta_def} vanish.
This limiting case is used in Section~\ref{sec:hcw_recovery}.

\begin{lemma}[Left Jacobian of $\mathrm{SO}(3)$: singular values and norms]
\label{lem:jacobian_norms}
With $\gamma_1 = (1-\cos\theta)/\theta^2$ and $\gamma_2 = (\theta-\sin\theta)/\theta^3$,
the left Jacobian on $\mathrm{SO}(3)$ is
$J_\ell^{\mathrm{SO}(3)}(\xi_R) = I + \gamma_1 S + \gamma_2 S^2$
\cite[\S7.1.5]{barfoot2024}. For $\theta \in (0,2\pi)$ its singular values are
\begin{equation}
\sigma\bigl(J_\ell^{\mathrm{SO}(3)}\bigr)
  = \left\{\,1,\ \frac{2\sin(\theta/2)}{\theta},\ \frac{2\sin(\theta/2)}{\theta}\,\right\},
\label{eq:jacobian_svals}
\end{equation}
which yield the operator norms
\begin{equation}
\|J_\ell^{\mathrm{SO}(3)}\| = 1,
\qquad
\|J_\ell^{-1,\mathrm{SO}(3)}\| = \frac{\theta/2}{\sin(\theta/2)}.
\label{eq:jacobian_norms}
\end{equation}
\end{lemma}

\begin{proof}
The closed form is standard~\cite{barfoot2024,sola2018}, and the derivation of
the singular values~\eqref{eq:jacobian_svals} is deferred to
\ref{app:jacobian_proof}. Since the singular values are strictly positive on
$(0,2\pi)$, $J_\ell^{\mathrm{SO}(3)}$ is invertible there. The operator
norms~\eqref{eq:jacobian_norms} are then a direct consequence: the spectral
norm equals the largest singular value, and the norm of the inverse the
reciprocal of the smallest.
\end{proof}

\subsection{Gravity Taylor Expansion}

Let $\rho := p - \bar{p}$. A second-order Taylor expansion of the two-body
gravitational field $g(q) = -\mu q/\|q\|^3$ ($\mu$ the gravitational parameter)
about $\bar{p}$ gives
\begin{equation}
    g(p) = g(\bar{p}) + G_I(t)\,\rho + f(\bar{p}, \rho),
    \label{eq:gravity_taylor}
\end{equation}
where
$G_I(t) = \tfrac{\mu}{r^3}\bigl(3\,\tfrac{\bar{p}\bar{p}^\top}{r^2} - I\bigr)$
is the classical gravity-gradient tensor~\cite{clohessy1960,schaub_junkins} and
\begin{equation}
    f(\bar{p}, \rho) := g(\bar{p} + \rho) - g(\bar{p}) - G_I(t)\rho
    = \int_0^1 (1-s)\, D^2 g(\bar{p} + s\rho)[\rho,\rho]\,ds
    \label{eq:gravity_remainder}
\end{equation}
is the exact integral remainder, where $D^2 g(q)$ is the Hessian of $g$,
regarded as a symmetric bilinear map, with induced norm $\|D^2 g(q)\| :=
\sup_{\|v\|=1}\|D^2 g(q)[v,v]\|$.

\begin{lemma}[Gravity linearization remainder]
\label{lem:gravity_remainder}
For the field $g$, the second-order Taylor
remainder~\eqref{eq:gravity_remainder} about $\bar{p}$, with $r = \|\bar{p}\|$
and $d = \|\rho\| < r$, satisfies the closed-form bound
\begin{equation}
    \|f(\bar{p},\rho)\| \le \frac{\mu\, d^2(3r - 2d)}{r^3(r-d)^2}.
    \label{eq:f_bound}
\end{equation}
\end{lemma}

\begin{proof}
\emph{Hessian norm.} Differentiating $g_i = -\mu q_i\|q\|^{-3}$ twice
(using $\partial_k\|q\|^{-m} = -m q_k\|q\|^{-m-2}$) gives, with $\hat q = q/\|q\|$,
\begin{equation}
    (D^2 g)_{ijk}
    = \frac{3\mu}{\|q\|^4}\bigl[\delta_{ik}\hat q_j + \delta_{jk}\hat q_i
      + \delta_{ij}\hat q_k - 5\hat q_i\hat q_j\hat q_k\bigr].
    \label{eq:hessian}
\end{equation}
Contracting with a unit vector $v$ in the last two slots and writing
$c_v = \hat q^\top v$ (so $v = c_v\hat q + v_\perp$, $v_\perp\perp\hat q$,
$\|v_\perp\|^2 = 1-c_v^2$),
\begin{equation}
    \|D^2 g(q)[v,v]\|^2
    = \frac{9\mu^2}{\|q\|^8}\bigl[(1-3c_v^2)^2 + 4c_v^2(1-c_v^2)\bigr]
    = \frac{9\mu^2}{\|q\|^8}\bigl(5c_v^4 - 2c_v^2 + 1\bigr).
    \label{eq:hessnormsq}
\end{equation}
On $c_v^2\in[0,1]$ the factor $5c_v^4 - 2c_v^2 + 1$ is maximized at $c_v^2 = 1$ (radial
$v$), value $4$, so $\|D^2 g(q)\| = 6\mu/\|q\|^4$.

\emph{Remainder integral.} Bounding the integrand
of~\eqref{eq:gravity_remainder} by this norm---a standard technique for
linearization-error bounds~\cite{khalil2002}---with $\|\rho\| = d$ and
$\|\bar{p}+s\rho\| \ge r-sd$ (valid for $s\in[0,1]$ since $d<r$),
\begin{equation}
    \|f(\bar{p},\rho)\| \le 6\mu d^2\int_0^1 \frac{1-s}{(r-sd)^4}\,ds
    = \frac{\mu\, d^2(3r-2d)}{r^3(r-d)^2},
    \label{eq:fint}
\end{equation}
which proves~\eqref{eq:f_bound}.
\end{proof}

\begin{theorem}[Gravity term linearization and error bounds]
\label{thm:gravity_linearization}
Let $\theta = \|\xi_R\| \in (0,\pi)$, $r = \|\bar{p}\|$, $d = \|\rho\| < r$,
and let $\varphi$ be the angle between $\xi_R$ and $\hat{r}$. The gravity term
\eqref{eq:gravity_block} admits the decomposition
\begin{equation}
    J_r^{-1}(\xi)\,\mathrm{Ad}^\vee_{X^{-1}}\tilde{m}
    =
    \begin{bmatrix} 0 \\ G(t)\,\xi_p \\ 0 \end{bmatrix}
    + \delta_{\mathrm{att}} + \delta_{\mathrm{grav}},
    \label{eq:gravity_linearized}
\end{equation}
where
\begin{align}
    \delta_{\mathrm{att}} &:=
    \begin{bmatrix} 0 \\ \Psi\,\xi_p \\ 0 \end{bmatrix}, \quad
    \Psi := J_\ell^{-1,\mathrm{SO}(3)}\,G(t)\,
    J_\ell^{\mathrm{SO}(3)} - G(t),
    \label{eq:delta_att} \\[4pt]
    \delta_{\mathrm{grav}} &:=
    \begin{bmatrix} 0 \\ J_\ell^{-1,\mathrm{SO}(3)}\,
    \bar{R}^\top f(\bar{p},\rho) \\ 0 \end{bmatrix},
    \label{eq:delta_grav}
\end{align}
and the remainder terms are bounded by
\begin{align}
    \|\delta_{\mathrm{att}}\| &\leq
    \frac{3\mu}{r^3}
    \left[
        \sin\!\left(\tfrac{\theta}{2}\right)\,\sin\varphi
        + \frac{\theta - \sin\theta}{4\sin(\theta/2)}\,\bigl|\sin 2\varphi\bigr|
    \right]\|\xi_p\|,
    \label{eq:bound_att} \\[4pt]
    \|\delta_{\mathrm{grav}}\| &\leq
    \frac{\theta/2}{\sin(\theta/2)}
    \cdot
    \frac{\mu\, d^2(3r - 2d)}{r^3(r-d)^2}.
    \label{eq:bound_grav}
\end{align}
Taking $\sin\varphi \le 1$ and $|\sin 2\varphi| \le 1$
in~\eqref{eq:bound_att} gives a bound in $\theta$ alone.
\end{theorem}

\begin{proof}
\textit{Decomposition.}
By~\eqref{eq:gravity_block} the gravity term has a single nonzero (velocity)
block, $J_r^{-1,\mathrm{SO}(3)}(\xi_R)\,R^\top\bigl(g(p)-g(\bar{p})\bigr)$. With
$\rho := p - \bar{p}$, the Taylor
expansion~\eqref{eq:gravity_taylor}--\eqref{eq:gravity_remainder} gives
$g(p) - g(\bar{p}) = G_I(t)\rho + f(\bar{p},\rho)$, and the $(1,3)$-block of
$\eta = \Exp(\xi^\wedge)$ gives $\rho = \bar{R}\,J_\ell^{\mathrm{SO}(3)}\xi_p$.
Writing $R = \bar{R}\eta_R$ with $\eta_R = \Exp(\xi_R^\wedge)$, so
$R^\top = \eta_R^{-1}\bar{R}^\top$, and applying the left/right identity
$J_r^{-1,\mathrm{SO}(3)}\eta_R^{-1} = J_\ell^{-1,\mathrm{SO}(3)}$
\cite{barfoot2024,sola2018},
\begin{align}
    J_r^{-1,\mathrm{SO}(3)}R^\top\bigl(g(p) - g(\bar{p})\bigr)
    &= J_r^{-1,\mathrm{SO}(3)}R^\top\bigl(G_I(t)\bar{R}J_\ell^{\mathrm{SO}(3)}\xi_p
       + f(\bar{p},\rho)\bigr) \nonumber\\
    &= J_r^{-1,\mathrm{SO}(3)}\eta_R^{-1}\bar{R}^\top
       \bigl(G_I(t)\bar{R}J_\ell^{\mathrm{SO}(3)}\xi_p + f(\bar{p},\rho)\bigr)
       \nonumber\\
    &= J_\ell^{-1,\mathrm{SO}(3)}\,G(t)\,J_\ell^{\mathrm{SO}(3)}\xi_p
       + J_\ell^{-1,\mathrm{SO}(3)}\bar{R}^\top f(\bar{p},\rho),
    \label{eq:decomp_chain}
\end{align}
using $G(t) = \bar{R}^\top G_I(t)\bar{R}$. The gravity term
$J_r^{-1}(\xi)\,\mathrm{Ad}^\vee_{X^{-1}}\tilde{m}$ decomposes to
\begin{equation}
    J_r^{-1}(\xi)\,\mathrm{Ad}^\vee_{X^{-1}}\tilde{m}
    =
    \begin{bmatrix} 0 \\ J_\ell^{-1,\mathrm{SO}(3)}\,G(t)\,
    J_\ell^{\mathrm{SO}(3)}\,\xi_p \\ 0 \end{bmatrix}
    +
    \begin{bmatrix} 0 \\ J_\ell^{-1,\mathrm{SO}(3)}\,
    \bar{R}^\top f(\bar{p},\rho) \\ 0 \end{bmatrix}.
    \label{eq:gravity_decomp}
\end{equation}
Adding and subtracting $G(t)\xi_p$ in the velocity block and writing
$J_\ell^{-1,\mathrm{SO}(3)} G(t) J_\ell^{\mathrm{SO}(3)} = G(t) + \Psi$
gives~\eqref{eq:gravity_linearized}--\eqref{eq:delta_grav}, where
$\|\rho\| \le \|\xi_p\|$ since $\|J_\ell^{\mathrm{SO}(3)}\| = 1$
(Lemma~\ref{lem:jacobian_norms}).

\textit{Bound on $\delta_{\mathrm{att}}$.}
Inserting $I = J_\ell^{-1,\mathrm{SO}(3)} J_\ell^{\mathrm{SO}(3)}$ into the
subtracted $G(t)$ gives $\Psi = J_\ell^{-1,\mathrm{SO}(3)}[G(t),
J_\ell^{\mathrm{SO}(3)}]$. Since $J_\ell^{\mathrm{SO}(3)} = I + \gamma_1 S + \gamma_2
S^2$ and $[G(t),I]=0$,
\begin{equation}
    \Psi = J_\ell^{-1,\mathrm{SO}(3)}
    \bigl(\gamma_1[G(t),S] + \gamma_2[G(t),S^2]\bigr).
    \label{eq:delta_commutator}
\end{equation}
The two commutators have the exact spectral norms
\begin{equation}
    \|[G(t),S]\| = \tfrac{3\mu}{r^3}\theta\sin\varphi,
    \qquad
    \|[G(t),S^2]\| = \tfrac{3\mu}{2r^3}\theta^2\,|\sin 2\varphi|,
    \label{eq:commutator_norms}
\end{equation}
derived in \ref{app:datt}. Applying submultiplicativity and the
triangle inequality to~\eqref{eq:delta_commutator} with
$\|J_\ell^{-1,\mathrm{SO}(3)}\| = (\theta/2)/\sin(\theta/2)$,
\begin{equation}
    \|\Psi\| \le \frac{\theta/2}{\sin(\theta/2)}
    \bigl(\gamma_1\|[G(t),S]\| + \gamma_2\|[G(t),S^2]\|\bigr).
    \label{eq:delta_assemble}
\end{equation}
Substituting~\eqref{eq:commutator_norms} with $\gamma_1\theta =
(1-\cos\theta)/\theta$ and $\gamma_2\theta^2 = (\theta-\sin\theta)/\theta$, and
simplifying the first term via $1-\cos\theta = 2\sin^2(\theta/2)$, gives the
bracket of~\eqref{eq:bound_att}; since $\delta_{\mathrm{att}} = \Psi\xi_p$,
$\|\delta_{\mathrm{att}}\| \le \|\Psi\|\,\|\xi_p\|$
yields~\eqref{eq:bound_att}.

\textit{Bound on $\delta_{\mathrm{grav}}$.}
By Lemma~\ref{lem:gravity_remainder}, $\|f(\bar{p},\rho)\| \le \mu
d^2(3r-2d)/[r^3(r-d)^2]$. Since $\bar{R}$ is orthogonal,
$\|\bar{R}^\top f\| = \|f\|$, and multiplying
by $\|J_\ell^{-1,\mathrm{SO}(3)}\| = (\theta/2)/\sin(\theta/2)$
gives~\eqref{eq:bound_grav}.
\end{proof}

\subsection{Linear Time-Varying Error System}
We now consider a linear time-varying error system for the matched-control case
$\tilde\nu \equiv 0$, so that $\nu = \bar\nu$ and the adjoint term
$-\ad_\nu$ of~\eqref{eq:loglinear} may be written with the reference input
$\bar\nu = (0,\bar a,\bar\omega)$; the general case $\tilde\nu \neq 0$ is treated
in Section~\ref{sec:control_mismatch}. Dropping the remainders
$\delta_{\mathrm{att}}$ and $\delta_{\mathrm{grav}}$
from~\eqref{eq:gravity_linearized} replaces the gravity term
$J_r^{-1}(\xi)\,\mathrm{Ad}^\vee_{X^{-1}}\tilde{m}$ by its linear part, whose
only nonzero (velocity) block is $G(t)\xi_p$. Substituting
into~\eqref{eq:loglinear} gives the linear time-varying (LTV) system
\begin{equation}
\dot{\xi}(t) = A(t)\,\xi(t),
\label{eq:LTV}
\end{equation}
with system matrix
\begin{equation}
A(t) =
\begin{bmatrix}
-[\bar{\omega}]_\times & I & 0 \\
G(t) & -[\bar{\omega}]_\times & -[\bar{a}]_\times \\
0 & 0 & -[\bar{\omega}]_\times
\end{bmatrix}.
\label{eq:A_matrix}
\end{equation}
Written componentwise,
\begin{subequations}
\begin{align}
\dot{\xi}_p &= \xi_v - \bar{\omega} \times \xi_p, \label{eq:xi_p_dot} \\
\dot{\xi}_v &= G(t)\,\xi_p - \bar{\omega} \times \xi_v - \bar{a} \times \xi_R,
\label{eq:xi_v_dot} \\
\dot{\xi}_R &= -\bar{\omega} \times \xi_R. \label{eq:xi_R_dot}
\end{align}
\label{eq:component_dynamics}
\end{subequations}
The state-transition matrix $\Phi(t,t_0)$ satisfies $\dot{\Phi} = A(t)\Phi$,
$\Phi(t_0,t_0) = I$, and is computed numerically for general thrusting
references.
 
\begin{remark}[Validity of the LTV approximation]
\label{rem:ltv_validity}
The approximation~\eqref{eq:LTV}--\eqref{eq:A_matrix} neglects the two
remainders bounded in Theorem~\ref{thm:gravity_linearization}, and is accurate
when both are small relative to the retained term $G(t)\xi_p$ (of magnitude up
to $2\mu\|\xi_p\|/r^3$). By~\eqref{eq:bound_att}, $\|\delta_{\mathrm{att}}\| =
O(\theta)\,\|\xi_p\|\cdot\mu/r^3$ vanishes with the attitude error
$\theta = \|\xi_R\|$; by
\eqref{eq:bound_grav}, $\|\delta_{\mathrm{grav}}\| = O(d^2/r^2)\cdot\mu/r^2$
vanishes with the relative separation $d = \|\rho\| \le \|\xi_p\|$. Both are
therefore negligible for small attitude error and small separation relative to
orbital radius, with quantitative thresholds obtained directly from
\eqref{eq:bound_att}--\eqref{eq:bound_grav}.
\end{remark}

\section{Recovery of HCW as a Special Case}
\label{sec:hcw_recovery}
The LTV system~\eqref{eq:LTV} recovers the classical Hill--Clohessy--Wiltshire
equations exactly under the standard HCW assumptions. No additional approximations
are introduced beyond those already stated in Remark~\ref{rem:ltv_validity}.

\begin{theorem}[HCW recovery]
\label{thm:hcw}
Consider a circular coasting reference orbit of radius $r_0$ with constant mean
motion $n = \sqrt{\mu/r_0^3}$ and angular velocity $\Omega = [0,0,n]^\top$.
Let the chief's body frame be the local-vertical local-horizontal (LVLH) frame,
with $\bar{b}_1$ radial, $\bar{b}_2$ along-track, and $\bar{b}_3$ along the
orbit normal, so that the body-frame radial unit vector of
Section~\ref{sec:gravity_linearization} is $\hat{r} = [1,0,0]^\top$.
Under the assumptions
\begin{equation}
    \bar{a} \equiv 0, \qquad \tilde{\nu} \equiv 0, \qquad \bar{\omega} = \Omega,
    \qquad \|\bar{p}\| = r_0 = \mathrm{const}, \qquad \xi_R(t_0) = 0,
    \label{eq:hcw_assumptions}
\end{equation}
the translational log-error dynamics~\eqref{eq:xi_p_dot}--\eqref{eq:xi_v_dot}
reduce to the classical Hill--Clohessy--Wiltshire equations
\begin{align}
    \ddot{\xi}_{p_x} - 2n\dot{\xi}_{p_y} - 3n^2\xi_{p_x} &= 0,
    \label{eq:hcw_x} \\
    \ddot{\xi}_{p_y} + 2n\dot{\xi}_{p_x} &= 0,
    \label{eq:hcw_y} \\
    \ddot{\xi}_{p_z} + n^2\xi_{p_z} &= 0,
    \label{eq:hcw_z}
\end{align}
where $\xi_R \equiv 0$ gives $J_\ell^{\mathrm{SO}(3)} = I$ (in the continuous
limit $\theta\to0$ noted in Section~\ref{sec:gravity_linearization}), so that
$\rho = p - \bar p =\bar{R}\,\xi_p$ exactly and
$\xi_p = [\xi_{p_x}, \xi_{p_y}, \xi_{p_z}]^\top$ is the relative position
resolved in the LVLH frame, matching the classical HCW coordinates.
\end{theorem}

\begin{proof}
\textit{Invariance of $\xi_R \equiv 0$.}
Under $\tilde{\nu} \equiv 0$ the attitude error obeys~\eqref{eq:xi_R_dot},
$\dot{\xi}_R = -\bar{\omega}\times\xi_R$, which is exact rather than linearized,
since the gravity term has no attitude component by Lemma~\ref{lem:gravblock},
and is linear and homogeneous in $\xi_R$. Hence $\xi_R(t_0) = 0$ implies $\xi_R(t) \equiv 0$ for all $t$: the
zero-attitude-error set is invariant, and the assumption is dynamically
consistent rather than an additional approximation.

With $\xi_R \equiv 0$ and $\bar{a}\equiv 0$, the $-[\bar{a}]_\times\xi_R$
coupling term in~\eqref{eq:xi_v_dot} vanishes and the dynamics reduce to
\begin{subequations}
\begin{align}
    \dot{\xi}_p &= \xi_v - \Omega \times \xi_p, \label{eq:inv_hcw_p} \\
    \dot{\xi}_v &= G\,\xi_p - \Omega \times \xi_v. \label{eq:inv_hcw_v}
\end{align}
\end{subequations}
\textit{Evaluation of $G$.}
With $\|\bar{p}\| = r_0$ and $\hat{r} = [1,0,0]^\top$, the body-frame tidal
tensor of Section~\ref{sec:gravity_linearization} is
\begin{equation}
    G = \frac{\mu}{r_0^3}\bigl(3\hat{r}\hat{r}^\top - I\bigr)
      = n^2\bigl(3\hat{r}\hat{r}^\top - I\bigr)
      = \begin{bmatrix} 2n^2 & 0 & 0 \\ 0 & -n^2 & 0 \\ 0 & 0 & -n^2 \end{bmatrix}.
    \label{eq:G_evaluated}
\end{equation}
\footnote{Note that $G \neq \mathrm{diag}(3n^2, 0, -n^2)$; the familiar HCW
tidal coefficients emerge only after combining $G$ with the centrifugal term
$\Omega\times(\Omega\times\xi_p)$ in the second-order form below.}

\textit{Second-order form.}
Differentiating~\eqref{eq:inv_hcw_p} and substituting~\eqref{eq:inv_hcw_v}:
\begin{equation}
    \ddot{\xi}_p = \dot{\xi}_v - \Omega\times\dot{\xi}_p
                 = G\xi_p - \Omega\times\xi_v - \Omega\times\dot{\xi}_p.
\end{equation}
Substituting $\xi_v = \dot{\xi}_p + \Omega\times\xi_p$
from~\eqref{eq:inv_hcw_p} into the $\Omega\times\xi_v$ term:
\begin{align}
    \ddot{\xi}_p
    &= G\xi_p - \Omega\times(\dot{\xi}_p + \Omega\times\xi_p)
       - \Omega\times\dot{\xi}_p \notag \\
    &= G\xi_p - 2\Omega\times\dot{\xi}_p - \Omega\times(\Omega\times\xi_p).
\end{align}
\textit{Scalar expansion.}
With $\Omega = [0,0,n]^\top$ and $\xi_p = [\xi_{p_x},\xi_{p_y},\xi_{p_z}]^\top$:
\begin{equation}
    2\Omega\times\dot{\xi}_p =
    \begin{bmatrix}-2n\dot{\xi}_{p_y}\\2n\dot{\xi}_{p_x}\\0\end{bmatrix},
    \qquad
    \Omega\times(\Omega\times\xi_p) =
    \begin{bmatrix}-n^2\xi_{p_x}\\-n^2\xi_{p_y}\\0\end{bmatrix}.
\end{equation}
Substituting~\eqref{eq:G_evaluated} and expanding componentwise:
\begin{align}
    \ddot{\xi}_{p_x} &= 2n^2\xi_{p_x} + 2n\dot{\xi}_{p_y} + n^2\xi_{p_x}
                      = 3n^2\xi_{p_x} + 2n\dot{\xi}_{p_y}, \notag\\
    \ddot{\xi}_{p_y} &= -n^2\xi_{p_y} - 2n\dot{\xi}_{p_x} + n^2\xi_{p_y}
                      = -2n\dot{\xi}_{p_x}, \notag\\
    \ddot{\xi}_{p_z} &= -n^2\xi_{p_z}.
\end{align}
Rearranging gives~\eqref{eq:hcw_x}--\eqref{eq:hcw_z}.
\end{proof}

\begin{remark}[Identical approximations]
\label{rem:hcw_exact}
The recovery in Theorem~\ref{thm:hcw} involves no attitude-induced
approximation: at $\xi_R \equiv 0$ the left Jacobian satisfies
$J_\ell^{\mathrm{SO}(3)} = I$, so $\Psi = 0$ and
$\delta_{\mathrm{att}} \equiv 0$ identically
by~\eqref{eq:delta_att}. The only term dropped in passing
from~\eqref{eq:loglinear} to~\eqref{eq:LTV} is $\delta_{\mathrm{grav}}$, the
second-order remainder of the gravity Taylor
expansion~\eqref{eq:gravity_remainder} --- precisely the term neglected in the
classical derivation of the Hill--Clohessy--Wiltshire equations. The two models
therefore make identical approximations, and~\eqref{eq:hcw_x}--\eqref{eq:hcw_z}
are an exact consequence of the $\mathrm{SE}_2(3)$ framework rather than a
separate linearization.
\end{remark}

\section{Control Mismatch Linearization}
\label{sec:control_mismatch}

Section~\ref{sec:gravity_linearization} linearized the gravity mismatch term
in~\eqref{eq:loglinear}, yielding a linear time-varying system valid under the
matched-control assumption $\tilde\nu \equiv 0$. This section relaxes that assumption. We
linearize the control mismatch term $J_\ell^{-1}(\xi)\,\tilde\nu$ to first order in $\xi$,
absorb the linear part into a modified system matrix, and bound the residual analytically.
The result is a forced linear time-varying system that handles arbitrary known control
mismatch, with matched control recovered as the special case $\tilde\nu \equiv 0$.

The mismatch vector $\tilde\nu = (0,\Delta a,\Delta\omega)^\top$ of~\eqref{eq:tildes} has
two physically distinct components,
\begin{equation}
    \Delta a := a - \bar a,
    \qquad
    \Delta\omega := \omega - \bar\omega.
    \label{eq:mismatch_components}
\end{equation}
The linearization treats these separately because they couple
into the dynamics through different blocks of $J_\ell^{-1}(\xi)$. The acceleration component
reduces to an $\SO$-only computation parallel to the gravity term of
Section~\ref{sec:gravity_linearization}; the angular velocity component activates additional
coupling blocks of $J_\ell^{-1}$ and requires a new bound.

\subsection{Block Structure of the Inverse Left Jacobian}
\label{sec:Jl_structure}

By Lemma~\ref{lem:jacobian_norms} the left Jacobian on $\SO$ is
$J_\ell^{\SO} = I + \gamma_1 S + \gamma_2 S^2$, and by \ref{app:jacobian_proof}
$S$ annihilates the axis $\mathrm{span}\{\hat\xi\}$ and satisfies $S^2 = -\theta^2 I$ on the
plane $P = \hat\xi^\perp$. Hence $S^3 = -\theta^2 S$, so $\mathrm{span}\{I,S,S^2\}$ is closed
under multiplication and the inverse may be sought in that span. Writing
$J_\ell^{-1,\SO} = I + k_1 S + k_2 S^2$, imposing $J_\ell^{-1,\SO}J_\ell^{\SO} = I$ and
reducing by $S^3 = -\theta^2 S$ gives the two scalar conditions
$\gamma_1 + k_1(1-\gamma_2\theta^2) - k_2\gamma_1\theta^2 = 0$ and
$\gamma_2 + k_1\gamma_1 + k_2(1-\gamma_2\theta^2) = 0$, whose solution is
\begin{equation}
    J_\ell^{-1,\SO}(\xi_R) = I - \tfrac{1}{2}S + c(\theta)\,S^2,
    \qquad
    c(\theta) := \frac{\beta(\theta)}{\theta^2},
    \qquad
    \beta(\theta) := 1 - \frac{\theta}{2}\cot\frac{\theta}{2},
    \label{eq:Jl_inv_so3}
\end{equation}
in agreement with~\cite[\S7.1.5]{barfoot2024}.

In the ordering $\xi = (\xi_p,\xi_v,\xi_R)$ the matrix representation of the algebra adjoint
is
\begin{equation}
    \ad_\xi =
    \begin{bmatrix}
        S & 0 & \skewmat{\xi_p} \\
        0 & S & \skewmat{\xi_v} \\
        0 & 0 & S
    \end{bmatrix},
    \label{eq:ad_xi_matrix}
\end{equation}
block upper-triangular with identical diagonal blocks $S$ and a vanishing $(1,2)$ block. By
the argument of Lemma~\ref{lem:gravblock}, every power of $\ad_\xi$, and hence every power
series in $\ad_\xi$, inherits this structure with diagonal blocks given by the corresponding
$\SO$ series. Applying this to $J_\ell(\xi) = \sum_{k\ge0}\ad_\xi^{\,k}/(k+1)!$, the diagonal
blocks are $J_\ell^{\SO}(\xi_R)$ and the $(1,3)$ and $(2,3)$ blocks are two instances of a
single function of a translation argument $u \in \mathbb{R}^3$. Since the $(1,3)$ block of
$\ad_\xi^{\,k}$ is $\sum_{i+j=k-1}S^i\skewmat{\xi_p}S^j$, weighting by $1/(k+1)!$ and
reindexing by $(i,j)$ identifies that function~\cite[Eq.~(7.86b)]{barfoot2024} as
\begin{equation}
    Q(u,\xi_R) := \sum_{i \ge 0}\ \sum_{j \ge 0}
    \frac{S^{i}\,\skewmat{u}\,S^{j}}{(i+j+2)!}\,.
    \label{eq:Q_series}
\end{equation}
Unlike Lemma~\ref{lem:gravblock}, where the gravity mismatch has only a velocity entry and
block upper-triangular multiplication extracts the diagonal alone, both off-diagonal blocks
are active here. Block-triangular inversion gives
\begin{equation}
    J_\ell^{-1}(\xi) =
    \begin{bmatrix}
        J_\ell^{-1,\SO} & 0 & -J_\ell^{-1,\SO} Q_p J_\ell^{-1,\SO} \\[3pt]
        0 & J_\ell^{-1,\SO} & -J_\ell^{-1,\SO} Q_v J_\ell^{-1,\SO} \\[3pt]
        0 & 0 & J_\ell^{-1,\SO}
    \end{bmatrix},
    \qquad
    \begin{aligned}
        Q_p &:= Q(\xi_p,\xi_R), \\
        Q_v &:= Q(\xi_v,\xi_R).
    \end{aligned}
    \label{eq:Jl_inv_block}
\end{equation}

Only three properties of $Q$ are used below. For any $u,u_1,u_2 \in \mathbb{R}^3$,
$\lambda_1,\lambda_2 \in \mathbb{R}$ and $A \in \SO$,
\begin{equation}
    Q(\lambda_1u_1 + \lambda_2u_2,\,\xi_R)
      = \lambda_1Q(u_1,\xi_R) + \lambda_2Q(u_2,\xi_R),
    \quad
    Q(u,0) = \tfrac{1}{2}\skewmat{u},
    \quad
    Q(Au,A\xi_R) = A\,Q(u,\xi_R)\,A^\top,
    \label{eq:Q_props}
\end{equation}
and $Q$ is the derivative of the left Jacobian along the translation argument,
\begin{equation}
    Q(u,\xi_R) = \frac{d}{dt}\bigg|_{t=0} J_\ell^{\SO}\bigl(\xi_R + t\,u\bigr).
    \label{eq:Q_deriv}
\end{equation}
Each follows directly from~\eqref{eq:Q_series}: the factor $\skewmat{u}$ occurs exactly once
in every term and the hat map is linear, giving linearity in $u$; only $i = j = 0$ survives at
$\xi_R = 0$, leaving $\skewmat{u}/2!$; the identity $A\skewmat{x}A^\top = \skewmat{Ax}$
conjugates each term $S^i\skewmat{u}S^j$, giving equivariance; and differentiating
$(S + t\skewmat{u})^k$ acts on one factor at a time and returns
$\sum_{i+j=k-1}S^i\skewmat{u}S^j$, so the same weighting recovers~\eqref{eq:Q_series} from the
series for $J_\ell^{\SO}$, establishing~\eqref{eq:Q_deriv}.

Applying~\eqref{eq:Jl_inv_block} to $\tilde\nu$ and separating the contributions from
$\Delta a$ and $\Delta\omega$:
\begin{equation}
    J_\ell^{-1}(\xi)\,\tilde\nu
    = \underbrace{
      \begin{pmatrix} 0 \\ J_\ell^{-1,\SO}\Delta a \\ 0 \end{pmatrix}}_{\text{from }\Delta a}
    + \underbrace{
      \begin{pmatrix}
        -J_\ell^{-1,\SO} Q_p J_\ell^{-1,\SO}\Delta\omega \\
        -J_\ell^{-1,\SO} Q_v J_\ell^{-1,\SO}\Delta\omega \\
        J_\ell^{-1,\SO}\Delta\omega
      \end{pmatrix}}_{\text{from }\Delta\omega}.
    \label{eq:Jl_inv_split}
\end{equation}
The acceleration mismatch acts only through the $\SO$ Jacobian, since the $(1,2)$ block
of~\eqref{eq:Jl_inv_block} vanishes; the angular velocity mismatch couples to position and
velocity through the $Q$-blocks.

\subsection{Linearization and Residual Bounds}
\label{sec:control_bounds}

The linearization is organized around the matrix-valued function
\begin{equation}
    \mathcal{M}(u,\xi_R) := -J_\ell^{-1,\SO}\,Q(u,\xi_R)\,J_\ell^{-1,\SO}
                  + \tfrac{1}{2}\skewmat{u},
    \qquad u \in \mathbb{R}^3,
    \label{eq:M_def}
\end{equation}
which measures the deviation of the $Q$-route from its zero-attitude limit. Since
$J_\ell^{-1,\SO} = I$ and $Q(u,0) = \tfrac12\skewmat{u}$ at $\xi_R = 0$
by~\eqref{eq:Jl_inv_so3} and~\eqref{eq:Q_props}, the subtracted term
in~\eqref{eq:M_def} is exactly the zero-attitude value of the product, so
$\mathcal{M}(u,0) = 0$.

\begin{lemma}[Induced norm of $\mathcal{M}$]
\label{lem:M_bound}
For $\theta \in (0,\pi)$ and all $u \in \mathbb{R}^3$,
\begin{equation}
    \|\mathcal{M}(u,\xi_R)\| \le \beta'(\theta)\,\|u\|,
    \label{eq:M_bound}
\end{equation}
with equality when $u$ is parallel to $\xi_R$.
\end{lemma}

\begin{proof}
\emph{Reduction to two directions.} By the equivariance in~\eqref{eq:Q_props}, together with
the corresponding properties of $J_\ell^{-1,\SO}$ and $\skewmat{\cdot}$,
$\mathcal{M}(Au,A\xi_R) = A\,\mathcal{M}(u,\xi_R)A^\top$ for any $A \in \SO$; since the
spectral norm is invariant under orthogonal
conjugation, $\|\mathcal{M}(u,\xi_R)\|$ depends only on $\|u\|$, $\theta$, and the angle between $u$
and $\xi_R$. Decompose $u$ with respect to the attitude-error axis as
\begin{equation}
    u = u_\parallel\,\hat\xi + u_\perp\,\hat w,
    \qquad
    u_\parallel := \hat\xi^\top u,
    \qquad
    u_\perp := \|u - u_\parallel\hat\xi\|,
    \qquad
    \hat w := \frac{u - u_\parallel\hat\xi}{u_\perp},
    \label{eq:u_decomp}
\end{equation}
so that $\hat w$ is a unit vector in the plane $P = \hat\xi^\perp$ of
\ref{app:jacobian_proof}, the triad $\{\hat\xi,\hat w,\hat\xi\times\hat w\}$ is that
of \ref{app:jacobian_proof}, and $\|u\|^2 = u_\parallel^2 + u_\perp^2$. If
$u_\perp = 0$, take $\hat w$ to
be any unit vector in $P$. By linearity of $\mathcal{M}$ in $u$,
$\mathcal{M}(u,\xi_R) = u_\parallel\,\mathcal{M}(\hat\xi,\xi_R)
+ u_\perp\,\mathcal{M}(\hat w,\xi_R)$.

\emph{Evaluation on the two directions.} Let $E := \skewmat{\hat\xi}$, so that $S = \theta E$
and $c\,\theta^2 = \beta$, giving $J_\ell^{-1,\SO} = I - \tfrac{\theta}{2}E + \beta(\theta)E^2$.
Set $\phi(t) := \xi_R + t\,u$. Since $\skewmat{u}$ is itself the hat of a vector,
$\skewmat{\phi(t)} = S + t\skewmat{u}$ remains skew for all $t$, so~\eqref{eq:Jl_inv_so3}
applies along the whole curve. Differentiating $J_\ell^{-1,\SO}J_\ell^{\SO} = I$ gives
$\tfrac{d}{dt}J_\ell^{-1,\SO} = -J_\ell^{-1,\SO}\bigl(\tfrac{d}{dt}J_\ell^{\SO}\bigr)J_\ell^{-1,\SO}$,
which with~\eqref{eq:Q_deriv} yields
\begin{equation}
    -J_\ell^{-1,\SO}\,Q(u,\xi_R)\,J_\ell^{-1,\SO}
    = \frac{d}{dt}\bigg|_{t=0}
      \Bigl[\, I - \tfrac{1}{2}\skewmat{\phi}
      + c\bigl(\|\phi\|\bigr)\skewmat{\phi}^2 \,\Bigr].
    \label{eq:M_derivative}
\end{equation}

For $u = \hat\xi$ one has $\phi(t) = (\theta+t)\hat\xi$ and
$\skewmat{\phi} = (\theta+t)E$, so~\eqref{eq:M_derivative} is an ordinary derivative in
$\theta$,
\begin{equation*}
    \frac{d}{d\theta}\Bigl[\, I - \tfrac{\theta}{2}E + \beta(\theta)E^2 \,\Bigr]
    = -\tfrac{1}{2}E + \beta'(\theta)\,E^2,
\end{equation*}
and since $\skewmat{\hat\xi} = E$ the linear term cancels against
$\tfrac12\skewmat{u}$ in~\eqref{eq:M_def}, leaving
\begin{equation}
    \mathcal{M}(\hat\xi,\xi_R) = m_a(\theta)\,E^2,
    \qquad
    m_a(\theta) := \beta'(\theta).
    \label{eq:m_a}
\end{equation}
For $u = \hat w$, write $P_w := \skewmat{\hat w}$, so
$\skewmat{\phi} = \theta E + t P_w$ and, by orthogonality of $\hat w$ and $\hat\xi$,
$\|\phi(t)\| = (\theta^2 + t^2)^{1/2}$ with $\tfrac{d}{dt}\|\phi\|\big|_{t=0} = 0$. The scalar
$c(\|\phi\|)$ is therefore stationary and only the matrix factors contribute:
\begin{equation*}
    \frac{d}{dt}\bigg|_{t=0}
    \Bigl[\, -\tfrac{1}{2}(\theta E + t P_w)
    + c(\|\phi\|)(\theta E + t P_w)^2 \,\Bigr]
    = -\tfrac{1}{2}P_w + c(\theta)\,\theta\,(E P_w + P_w E).
\end{equation*}
The linear term again cancels against $\tfrac12\skewmat{u}$, and evaluating the
anticommutator by
$\skewmat{x}\skewmat{y} + \skewmat{y}\skewmat{x} = y x^\top + x y^\top - 2(x^\top y)I$ at the
orthogonal pair $(\hat\xi,\hat w)$ gives
\begin{equation}
    \mathcal{M}(\hat w,\xi_R)
      = m_p(\theta)\bigl(\hat\xi\,\hat w^\top + \hat w\,\hat\xi^\top\bigr),
    \qquad
    m_p(\theta) := \theta\,c(\theta) = \frac{\beta(\theta)}{\theta}.
    \label{eq:m_p}
\end{equation}

\emph{Assembly.} Both~\eqref{eq:m_a} and~\eqref{eq:m_p} are symmetric, hence so is
$\mathcal{M}(u,\xi_R)$, and its spectral norm equals its largest eigenvalue magnitude. Since
$E^2 = \mathrm{diag}(0,-1,-1)$ in the triad of~\eqref{eq:u_decomp},
\begin{equation}
    \mathcal{M}(u,\xi_R) =
    \begin{bmatrix}
        0 & u_\perp\,m_p & 0 \\
        u_\perp\,m_p & -u_\parallel\,m_a & 0 \\
        0 & 0 & -u_\parallel\,m_a
    \end{bmatrix}.
    \label{eq:M_matrix}
\end{equation}
The $(\hat\xi\times\hat w)$ entry is decoupled, with eigenvalue magnitude
$|u_\parallel|m_a$. The
leading $2\times2$ block has characteristic polynomial
$z^2 + u_\parallel\,m_a z - u_\perp^2 m_p^2 = 0$, whose larger root in magnitude dominates the
decoupled contribution, yielding
\begin{equation}
    \|\mathcal{M}(u,\xi_R)\|
    = \tfrac{1}{2}\Bigl(|u_\parallel|\,m_a
      + \sqrt{u_\parallel^2 m_a^2 + 4 u_\perp^2 m_p^2}\Bigr).
    \label{eq:M_norm}
\end{equation}

\emph{Maximization over direction.} Set $\|u\| = 1$ and $\lambda := u_\parallel^2 \in [0,1]$, so
$u_\perp^2 = 1-\lambda$ and~\eqref{eq:M_norm} becomes
\begin{equation*}
    \|\mathcal{M}\| = \tfrac{1}{2}\left(\sqrt{\lambda}\,m_a
    + \sqrt{\lambda\,(m_a^2 - 4m_p^2) + 4m_p^2}\,\right).
\end{equation*}
Both $m_a$ and $m_p$ are positive on $(0,\pi)$, since $\beta > 0$ there gives $m_p > 0$ and
\ref{app:ma_mp} gives $m_a > 2m_p$. The first term is therefore increasing in
$\lambda$; the radicand of the second is affine in $\lambda$ with slope $m_a^2 - 4m_p^2$,
non-negative precisely when $m_a \ge 2m_p$, so the second term is increasing in $\lambda$ as
well. The maximum over $[0,1]$ is attained at $\lambda = 1$, that is at $u$ parallel to
$\xi_R$, where $u_\perp = 0$ and $\|\mathcal{M}\| = m_a\|u\| = \beta'(\theta)\|u\|$.
\end{proof}

\begin{theorem}[Control mismatch linearization]
\label{thm:control_mismatch}
For $\theta = \|\xi_R\| \in (0,\pi)$, the action of $J_\ell^{-1}(\xi)$ on $\tilde\nu$ admits
the exact decomposition
\begin{equation}
    J_\ell^{-1}(\xi)\,\tilde\nu
    = \tilde\nu + L(\tilde\nu)\,\xi
    + \begin{pmatrix}
        \varepsilon_p(\xi,t) \\[2pt]
        \varepsilon_v(\xi,t) \\[2pt]
        \varepsilon_R(\xi,t)
      \end{pmatrix},
    \label{eq:Jl_tilde_nu_decomp}
\end{equation}
where the linear-in-$\xi$ correction matrix is
\begin{equation}
    L(\tilde\nu) := \ad_{\frac{1}{2}\tilde\nu}
    = \tfrac{1}{2}
    \begin{bmatrix}
        \skewmat{\Delta\omega} & 0 & 0 \\[2pt]
        0 & \skewmat{\Delta\omega} & \skewmat{\Delta a} \\[2pt]
        0 & 0 & \skewmat{\Delta\omega}
    \end{bmatrix},
    \label{eq:L_def}
\end{equation}
and the residual blocks $\varepsilon_p,\varepsilon_v,\varepsilon_R \in \mathbb{R}^3$ satisfy
\begin{align}
    \|\varepsilon_p(\xi,t)\| &\le \alpha(\theta)\,\|\xi_p\|\,\|\Delta\omega(t)\|,
    \label{eq:epsp_bound}\\
    \|\varepsilon_v(\xi,t)\| &\le \beta(\theta)\,\|\Delta a(t)\|
        + \alpha(\theta)\,\|\xi_v\|\,\|\Delta\omega(t)\|,
    \label{eq:epsv_bound}\\
    \|\varepsilon_R(\xi,t)\| &\le \beta(\theta)\,\|\Delta\omega(t)\|,
    \label{eq:epsR_bound}
\end{align}
with
\begin{equation}
    \beta(\theta) := 1 - \frac{\theta}{2}\cot\frac{\theta}{2}
    = \frac{\theta^2}{12} + O(\theta^4),
    \qquad
    \alpha(\theta) := \beta'(\theta)
    = \frac{\theta}{4}\csc^2\frac{\theta}{2} - \frac{1}{2}\cot\frac{\theta}{2}
    = \frac{\theta}{6} + O(\theta^3).
    \label{eq:alpha_beta_def}
\end{equation}
\end{theorem}

\begin{proof}
\emph{Diagonal contributions.} By~\eqref{eq:Jl_inv_so3}, for any $z \in \mathbb{R}^3$,
\begin{equation}
    J_\ell^{-1,\SO} z = z - \tfrac{1}{2}S z + \varepsilon_a(\xi_R,z),
    \qquad
    \varepsilon_a(\xi_R,z) := c(\theta)\,S^2 z,
    \label{eq:Jz_lin}
\end{equation}
which is exact, the closed form~\eqref{eq:Jl_inv_so3} terminating at second order. By the
cross-product identity $x \times y = -\,y \times x$ applied to $\xi_R \times z$, the linear
correction rewrites as $-\tfrac12 S z = \tfrac12\skewmat{z}\xi_R$, so that it is linear in the
state with the input as coefficient. Applying~\eqref{eq:Jz_lin} with $z = \Delta a$ gives the
velocity-block expansion and with $z = \Delta\omega$ the rotation-block expansion.

By \ref{app:jacobian_proof}, $S^2$ vanishes on the axis and equals $-\theta^2 I$ on
$P = \hat\xi^\perp$, so $\|S^2\| = \theta^2$, attained on $P$. Hence
\begin{equation}
    \|\varepsilon_a(\xi_R,z)\| \le c(\theta)\,\theta^2\,\|z\| = \beta(\theta)\,\|z\|,
    \label{eq:epsa_bound}
\end{equation}
with equality for $z \in P$. Applying with $z = \Delta a$ gives the first term
of~\eqref{eq:epsv_bound}, and with $z = \Delta\omega$ gives~\eqref{eq:epsR_bound}.

\emph{Off-diagonal contributions.} By~\eqref{eq:M_def}, the position and velocity blocks of
the $\Delta\omega$ column of~\eqref{eq:Jl_inv_split} are
\begin{equation}
    -J_\ell^{-1,\SO}Q(u,\xi_R)J_\ell^{-1,\SO}\Delta\omega
    = -\tfrac{1}{2}\skewmat{u}\Delta\omega + \mathcal{M}(u,\xi_R)\Delta\omega
    = \tfrac{1}{2}\skewmat{\Delta\omega}u + \mathcal{M}(u,\xi_R)\Delta\omega,
    \label{eq:Dw_pv_expansion}
\end{equation}
again by antisymmetry of the cross product. Setting $u = \xi_p$ and $u = \xi_v$ and
collecting all linear-in-$\xi$ terms into $L(\tilde\nu)$ gives the
decomposition~\eqref{eq:Jl_tilde_nu_decomp} with the explicit forms
\begin{equation}
    \varepsilon_p = \mathcal{M}(\xi_p,\xi_R)\Delta\omega,
    \qquad
    \varepsilon_v = \varepsilon_a(\xi_R,\Delta a) + \mathcal{M}(\xi_v,\xi_R)\Delta\omega,
    \qquad
    \varepsilon_R = \varepsilon_a(\xi_R,\Delta\omega).
    \label{eq:eps_components}
\end{equation}
Lemma~\ref{lem:M_bound} applied with $u = \xi_p$ gives~\eqref{eq:epsp_bound}, and with
$u = \xi_v$ the second term of~\eqref{eq:epsv_bound}.

\emph{Form of $L$.} The collected linear terms are $\tfrac12\skewmat{\Delta\omega}\xi_p$ in
the position block, $\tfrac12\skewmat{\Delta\omega}\xi_v + \tfrac12\skewmat{\Delta a}\xi_R$ in
the velocity block, and $\tfrac12\skewmat{\Delta\omega}\xi_R$ in the rotation block. Comparing
with~\eqref{eq:ad_xi_matrix} evaluated at
$\tfrac12\tilde\nu = (0,\tfrac12\Delta a,\tfrac12\Delta\omega)$ identifies this matrix as
$\ad_{\frac12\tilde\nu}$, which is~\eqref{eq:L_def}. The small-angle orders
in~\eqref{eq:alpha_beta_def} follow from the expansion of $\cot$.
\end{proof}

\subsection{Forced Linear Time-Varying Error System}
\label{sec:control_dynamics}

\begin{proposition}[Forced log-error dynamics]
\label{prop:forced_LTV}
Let $X(t)$ and $\bar X(t)$ evolve under~\eqref{eq:eom-actual} and its overlined counterpart,
and let $\theta = \|\xi_R\| \in (0,\pi)$. Then the log-error~\eqref{eq:xidef} satisfies
\begin{equation}
    \dot\xi = \tilde A(t)\,\xi + \tilde\nu(t)
    + \varepsilon(\xi,t) + \delta_{\mathrm{att}} + \delta_{\mathrm{grav}},
    \label{eq:dynamics_full}
\end{equation}
where $\varepsilon = (\varepsilon_p,\varepsilon_v,\varepsilon_R)^\top$ is the residual of
Theorem~\ref{thm:control_mismatch}, $\delta_{\mathrm{att}}$ and $\delta_{\mathrm{grav}}$ are
the remainders of Theorem~\ref{thm:gravity_linearization}, and
\begin{equation}
    \tilde A(t) := A(t) - L(\tilde\nu(t))
    =
    \begin{bmatrix}
        -\skewmat{\omega_{\mathrm m}} & I & 0 \\[2pt]
        G(t) & -\skewmat{\omega_{\mathrm m}} & -\skewmat{a_{\mathrm m}} \\[2pt]
        0 & 0 & -\skewmat{\omega_{\mathrm m}}
    \end{bmatrix},
    \qquad
    \begin{aligned}
        a_{\mathrm m} &:= \tfrac{1}{2}(a + \bar a), \\
        \omega_{\mathrm m} &:= \tfrac{1}{2}(\omega + \bar\omega).
    \end{aligned}
    \label{eq:Atilde_block}
\end{equation}
Dropping the three bounded remainders gives the forced linear time-varying system
\begin{equation}
    \dot\xi(t) = \tilde A(t)\,\xi(t) + \tilde\nu(t),
    \label{eq:forced_LTV}
\end{equation}
which reduces to the matched-control system~\eqref{eq:LTV} when $\tilde\nu \equiv 0$. In
particular the attitude row of $\tilde A(t)$ vanishes in its position and velocity block
columns, so the attitude error remains autonomous.
\end{proposition}

\begin{proof}
Substituting the decomposition~\eqref{eq:Jl_tilde_nu_decomp} of
Theorem~\ref{thm:control_mismatch} and the gravity
linearization~\eqref{eq:gravity_linearized} of Theorem~\ref{thm:gravity_linearization}
into~\eqref{eq:loglinear} gives~\eqref{eq:dynamics_full} with
$\tilde A = A_C - \ad_\nu + L(\tilde\nu)$, carrying $G(t)$ in the velocity--position block.

For the first form of~\eqref{eq:Atilde_block}, note that $A(t)$ of~\eqref{eq:A_matrix} is
written with the reference input $\bar\nu$, whereas~\eqref{eq:loglinear} carries $-\ad_\nu$,
formed from the deputy input. Since $\ad$ is linear in its subscript,
$\ad_\nu = \ad_{\bar\nu} + \ad_{\tilde\nu}$, and combining with
$L(\tilde\nu) = \ad_{\frac12\tilde\nu}$ from~\eqref{eq:L_def} leaves
$-\ad_{\bar\nu} - \ad_{\frac12\tilde\nu}$, that is, $A(t) - L(\tilde\nu)$.

For the second form, subtract~\eqref{eq:L_def} from~\eqref{eq:A_matrix} block by block. The
angular-velocity blocks become $-\skewmat{\bar\omega + \frac12\Delta\omega}$ and the
thrust--attitude block $-\skewmat{\bar a + \frac12\Delta a}$; since
$\Delta\omega = \omega - \bar\omega$ and $\Delta a = a - \bar a$
by~\eqref{eq:mismatch_components}, these equal $-\skewmat{\omega_{\mathrm m}}$ and
$-\skewmat{a_{\mathrm m}}$.

Finally, $\tilde\nu \equiv 0$ gives $L(\tilde\nu) = 0$ by~\eqref{eq:L_def}, each component of
$\varepsilon$ vanishes by~\eqref{eq:eps_components}, and $a_{\mathrm m} = \bar a$,
$\omega_{\mathrm m} = \bar\omega$, so~\eqref{eq:Atilde_block} reduces
to~\eqref{eq:A_matrix}.
\end{proof}

\begin{remark}[Mean-input form]
\label{rem:mean_input}
Writing $\nu_{\mathrm m} := \tfrac12(\nu + \bar\nu) = (0,a_{\mathrm m},\omega_{\mathrm m})$,
the two identities of the proof combine into
\begin{equation}
    \tilde A(t) = A_C - \ad_{\nu_{\mathrm m}(t)} + G\text{-block},
    \label{eq:Atilde_mean}
\end{equation}
which is~\eqref{eq:A_matrix} with $\bar\nu$ replaced by $\nu_{\mathrm m}$. Absorbing the
first-order Jacobian correction therefore does not add structure to the system matrix; it
re-centers the linearization at the mean of the two vehicles' inputs rather than at the
reference's. The matched-control case $\nu = \bar\nu$ recovers $\nu_{\mathrm m} = \bar\nu$
immediately.
\end{remark}

\begin{remark}[Validity of the linearization]
\label{rem:validity_control}
The forced LTV system~\eqref{eq:forced_LTV} requires $\theta = \|\xi_R\|$ small enough that
$\beta(\theta)$ and $\alpha(\theta)$ are small. The restriction $\theta < \pi$ avoids the
$\SO$ exponential coordinate singularity, and is benign for proximity operations where
$\theta$ is typically a few degrees; the coupling correction is in fact validated out to
$\theta = 20^\circ$ in Section~\ref{sec:num_attitude}, where it still removes more than
$91\%$ of the coupling-blind terminal error in low Earth orbit. Combined with Remark~\ref{rem:ltv_validity}, the
assumptions for~\eqref{eq:forced_LTV} are $\theta \ll 1$ (control linearization and
attitude-induced gravity) and $\|\xi_p\|/r \ll 1$ (higher-order gravity). A sustained
angular-velocity mismatch is a maneuver rather than a steady state: it drives the
relative rotation toward $\pi$, where the logarithmic coordinate saturates and this
regime is left (Section~\ref{sec:num_mismatch}).

By~\eqref{eq:alpha_beta_def}, $\beta(\theta) = O(\theta^2)$ and $\alpha(\theta) = O(\theta)$,
so all three residuals vanish with the attitude error. The three blocks carry different
physical units, however, so each must be compared with the retained term in its own row
of~\eqref{eq:Jl_tilde_nu_decomp}: $\tfrac12\skewmat{\Delta\omega}\xi_p$ for
$\varepsilon_p$, $\Delta a$ for $\varepsilon_v$, and $\Delta\omega$ for $\varepsilon_R$.
Because $\varepsilon_p$ alone scales with the separation $\|\xi_p\|$, it is the dominant
residual at proximity-operations scales. Quantitative thresholds follow directly
from~\eqref{eq:epsp_bound}--\eqref{eq:epsR_bound}. The numerical study of
Section~\ref{sec:results} exercises both the matched-control case, for which
$\varepsilon \equiv 0$, and mismatched control, where Section~\ref{sec:num_mismatch}
validates all three residual bounds along the transfer.
\end{remark}

\begin{remark}[Sharpness]
\label{rem:sharpness_control}
Neither constant in~\eqref{eq:alpha_beta_def} can be reduced. Equality holds
in~\eqref{eq:epsR_bound} whenever $\Delta\omega$ lies in the plane $P = \hat\xi^\perp$, and
in~\eqref{eq:epsp_bound} whenever $\xi_p$ is parallel to $\xi_R$ and $\Delta\omega$ also lies in $P$, so no smaller $\beta$ or
$\alpha$ satisfies the bounds for all admissible states and inputs. Unlike the bound~\eqref{eq:bound_att} of
Theorem~\ref{thm:gravity_linearization}, which retains the angle $\varphi$ between $\xi_R$ and
$\hat r$, the bounds~\eqref{eq:epsp_bound}--\eqref{eq:epsR_bound} are stated with the direction
dependence already maximized: $\varphi$ is computable along a known reference trajectory,
whereas the orientation of $\xi_p$ and $\xi_v$ relative to $\xi_R$ is not known in advance,
and the worst case costs only a factor of two since $m_a = 2m_p$ to leading order. The
elementary bound $\alpha(\theta) \le \tfrac13\tan(\theta/2)$ holds on $(0,\pi)$, the two
agreeing to within $0.04\%$ at $\theta = 5^\circ$.
\end{remark}

%

\section{Two-Impulse $\Delta v$ Planning}
\label{sec:dv_planning}

Sections~\ref{sec:gravity_linearization}--\ref{sec:control_mismatch} produced a
forced linear time-varying system, \eqref{eq:forced_LTV}, together with
certified bounds on everything discarded in reaching it. This section derives a
two-impulse rendezvous planner from that system. Because the rendezvous condition
is exact in log coordinates, the only approximation entering the planner is in
propagation, so the bounds of Theorems~\ref{thm:gravity_linearization}
and~\ref{thm:control_mismatch} apply to it directly. The resulting planner
differs from classical formulations in a single term: an attitude--translation
coupling contribution that vanishes whenever both vehicles coast, and that
dominates the terminal error when the coupling ratio of Remark~\ref{rem:coupling}
exceeds unity.

The derivation is carried out once, over the general system~\eqref{eq:forced_LTV};
matched control, coasting references, and HCW follow as special cases in
Sec.~\ref{sec:dv_mismatch}. The planner requires $\bar R(t)$, $\bar a(t)$, and
$\bar\omega(t)$ along the reference, and---when the control is mismatched---the
deputy's own input profile $\nu(t)$ over the transfer horizon, since $\tilde A$
is built from the mean inputs $\nu_{\mathrm m}$ of
Remark~\ref{rem:mean_input}. This is available in cooperative proximity
operations, where the maneuver plans are known by construction.

\subsection{Problem Statement}
\label{sec:dv_problem}

Given the log-error state $\xi(t_0)$ and a transfer horizon $T = t_0 + \Delta t$,
find impulsive velocity corrections $\Delta\xi_v(t_0)$ and $\Delta\xi_v(T)$ such
that
\begin{equation}
    \xi_p(T) = 0,
    \qquad
    \xi_v(T^+) = 0 .
    \label{eq:rendezvous_condition}
\end{equation}

Conditions~\eqref{eq:rendezvous_condition} are stated in log coordinates but are
equivalent to physical rendezvous without approximation. The $(1,3)$ and $(1,2)$
blocks of $\eta = \Exp(\xi^\wedge)$ give
\begin{equation}
    \rho = p - \bar p = \bar R\, J_\ell^{\SO}(\xi_R)\,\xi_p,
    \qquad
    v - \bar v = \bar R\, J_\ell^{\SO}(\xi_R)\,\xi_v,
    \label{eq:rho_xi_relation}
\end{equation}
and both factors are invertible: $\bar R \in \SO$ always, and
$J_\ell^{\SO}(\xi_R)$ for $\theta \in (0,2\pi)$ by
Lemma~\ref{lem:jacobian_norms}, a range implied by the injectivity restriction
$\theta < \pi$ in force throughout. The maps
$\xi_p \mapsto \rho$ and $\xi_v \mapsto v - \bar v$ are therefore linear
bijections, so $\xi_p(T) = 0 \iff p(T) = \bar p(T)$ and
$\xi_v(T^+) = 0 \iff v(T^+) = \bar v(T)$.

\subsection{State Transition Matrix and Forced Response}
\label{sec:dv_stm}

Let $\tilde\Phi(t,t_0)$ denote the STM of the homogeneous part
of~\eqref{eq:forced_LTV}, satisfying
$\dot{\tilde\Phi} = \tilde A(t)\tilde\Phi$, $\tilde\Phi(t_0,t_0) = I$, and
partitioned conformally with $\xi = (\xi_p,\xi_v,\xi_R)$ as
\begin{equation}
\tilde\Phi =
\begin{bmatrix}
\tilde\Phi_{pp} & \tilde\Phi_{pv} & \tilde\Phi_{pR} \\
\tilde\Phi_{vp} & \tilde\Phi_{vv} & \tilde\Phi_{vR} \\
\tilde\Phi_{Rp} & \tilde\Phi_{Rv} & \tilde\Phi_{RR}
\end{bmatrix}.
\label{eq:STM_partition}
\end{equation}
The known forcing $\tilde\nu(t)$ contributes the particular integral
\begin{equation}
b(T,t_0) := \int_{t_0}^{T} \tilde\Phi(T,\tau)\,\tilde\nu(\tau)\,d\tau,
\qquad b = (b_p,\,b_v,\,b_R)^\top ,
\label{eq:forcing_integral}
\end{equation}
so that the free response of~\eqref{eq:forced_LTV} from $\xi(t_0)$ is
$\xi(T) = \tilde\Phi(T,t_0)\,\xi(t_0) + b(T,t_0)$.

The sparsity of $\tilde A(t)$ transfers to $\tilde\Phi$. Since the attitude row
of~\eqref{eq:Atilde_block} vanishes in its position and velocity block columns,
\begin{equation}
    \tilde\Phi_{Rp} = \tilde\Phi_{Rv} = 0,
    \qquad
    \xi_R(T) = \tilde\Phi_{RR}(T,t_0)\,\xi_R(t_0) + b_R(T,t_0),
    \label{eq:attitude_autonomous}
\end{equation}
which is computable in advance of any impulse selection. Conversely, for
thrusting references the term $-\skewmat{a_{\mathrm m}}\xi_R$ in the velocity row
propagates attitude error into translation, so $\tilde\Phi_{pR} \neq 0$ and
$\tilde\Phi_{vR} \neq 0$ in general. When $a_{\mathrm m} \equiv 0$ these blocks
vanish and the translational and rotational subsystems decouple completely.

Computationally, $\tilde\Phi$ and $b$ are obtained from a single joint
integration,
\begin{equation}
\frac{d}{dt}
\begin{bmatrix} \tilde{\Phi}(t,t_0) \\ b(t,t_0) \end{bmatrix}
=
\begin{bmatrix}
\tilde{A}(t)\,\tilde{\Phi}(t,t_0) \\
\tilde{A}(t)\,b(t,t_0) + \tilde{\nu}(t)
\end{bmatrix},
\qquad
\tilde{\Phi}(t_0,t_0) = I,\quad b(t_0,t_0) = 0,
\label{eq:joint_integration}
\end{equation}
adding nine scalar states to the $81$-state STM ODE.

Both $\tilde\Phi$ and $b$ depend only on the reference arc and the deputy's
planned input profile, not on the log-error state, so for a known maneuver plan
they may be computed once and reused across replans from updated navigation
solutions. Conversely, any change to the planned inputs of either vehicle---a
chief thrust update, or a deputy reorientation that alters $\omega$---changes
$\nu_{\mathrm m}$ and hence $\tilde A$, requiring reintegration.

\subsection{Impulse Structure in Log Coordinates}
\label{sec:dv_impulse}

An impulsive maneuver changes the inertial velocity discontinuously while
leaving position and attitude fixed. The following lemma records how such a jump
appears in the log-error state, and---in the same stroke---the exact map back to
a physical velocity increment.

\begin{lemma}[Impulse action on the log-error state]
\label{lem:impulse_action}
Let the deputy receive an impulsive inertial velocity increment
$\delta v \in \mathbb{R}^3$ at time $t_0$, with $\theta = \|\xi_R(t_0)\| < \pi$.
Then $\xi_p$ and $\xi_R$ are continuous across the impulse, and
\begin{equation}
    \xi_v(t_0^+) = \xi_v(t_0^-) + \Delta\xi_v(t_0),
    \qquad
    \Delta\xi_v(t_0) = J_\ell^{-1,\SO}\bigl(\xi_R(t_0)\bigr)\,\bar R(t_0)^\top\,\delta v .
    \label{eq:impulse}
\end{equation}
Conversely, the physical increment realizing a commanded $\Delta\xi_v$ is
\begin{equation}
    \delta v = \bar R\, J_\ell^{\SO}(\xi_R)\,\Delta\xi_v
    \quad (\text{inertial frame}),
    \qquad
    \delta v^{\mathcal B} = J_r^{\SO}(\xi_R)\,\Delta\xi_v
    \quad (\text{deputy body frame}),
    \label{eq:impulse_map}
\end{equation}
both exact.
\end{lemma}

\begin{proof}
The blocks of $\eta = \bar X^{-1}X$ are
$\eta_R = \bar R^\top R$, $\eta_v = \bar R^\top(v - \bar v)$, and
$\eta_p = \bar R^\top(p - \bar p)$; the blocks of $\eta = \Exp(\xi^\wedge)$ are
$\eta_R = \Exp(\xi_R^\wedge)$, $\eta_v = J_\ell^{\SO}(\xi_R)\xi_v$, and
$\eta_p = J_\ell^{\SO}(\xi_R)\xi_p$, the latter two sharing a single Jacobian
because the two translational columns of~\eqref{eq:xiblock} enter identically.
An impulse leaves $p$ and $R$ unchanged, hence leaves $\eta_p$ and $\eta_R$
unchanged, hence leaves $\xi_R$ and---since $J_\ell^{\SO}$ depends on $\xi_R$
alone and is invertible---leaves $\xi_p$ unchanged. It shifts
$\eta_v$ by $\bar R^\top\delta v$, giving~\eqref{eq:impulse} on applying
$J_\ell^{-1,\SO}$. Inverting yields the first expression
in~\eqref{eq:impulse_map}. For the second, write $R = \bar R\,\eta_R$, so
$\delta v^{\mathcal B} = R^\top \delta v = \eta_R^{-1}J_\ell^{\SO}\Delta\xi_v$;
the left/right identity $J_r^{-1,\SO}\eta_R^{-1} = J_\ell^{-1,\SO}$ used in the
proof of Theorem~\ref{thm:gravity_linearization} rearranges to
$\eta_R^{-1}J_\ell^{\SO} = J_r^{\SO}$.
\end{proof}

The left Jacobian appears in~\eqref{eq:impulse_map} for the same reason it
appears in~\eqref{eq:rho_xi_relation}: log coordinates and physical increments
differ by $J_\ell^{\SO}$, a distortion that is exact and computable, not an
error to be bounded. Retaining it costs one $3\times3$ product and removes the
$O(\theta)$ mapping error that an identification $\delta v = \bar R\Delta\xi_v$
would incur. By Lemma~\ref{lem:jacobian_norms} the change in magnitude is a mild
contraction,
\begin{equation}
    \frac{2\sin(\theta/2)}{\theta}\,\|\Delta\xi_v\|
    \;\le\; \|\delta v\| \;\le\; \|\Delta\xi_v\| ,
    \label{eq:impulse_norm_bounds}
\end{equation}
with the two sides agreeing to $O(\theta^2)$. The
bound~\eqref{eq:impulse_norm_bounds} constrains the magnitude alone, however,
and the magnitude is not where the distortion lies. On the plane
$\hat\xi^\perp$ the map $J_\ell^{\SO}$ is multiplication by the complex number
$c_{\mathbb C}$ of~\eqref{eq:app_c}, whose argument is exactly $\theta/2$: the
naive identification misdirects the impulse by a rotation of $\theta/2$, an
identity rather than a leading-order result, while shortening it by the
second-order factor $2\sin(\theta/2)/\theta = 1 - \theta^2/24 + O(\theta^4)$.
The transverse error therefore exceeds the magnitude error by $12/\theta$, a
factor of $137$ at $\theta = 5^\circ$, and Section~\ref{sec:num_mechanism}
measures the consequence for the planner.

\subsection{Departure and Arrival Impulses}
\label{sec:dv_solution}

\begin{proposition}[Two-impulse rendezvous solution]
\label{prop:two_impulse}
Consider the forced system~\eqref{eq:forced_LTV} with $\tilde\Phi$ and $b$
as in~\eqref{eq:STM_partition}--\eqref{eq:forcing_integral}, and suppose
$\tilde\Phi_{pv}(T,t_0)$ is invertible. Then the
conditions~\eqref{eq:rendezvous_condition} are met by a unique impulse pair. The
propagated position is
\begin{equation}
\xi_p(T) = \tilde\Phi_{pp}\,\xi_p(t_0)
+ \tilde\Phi_{pv}\bigl[\xi_v(t_0)+\Delta\xi_v(t_0)\bigr]
+ \tilde\Phi_{pR}\,\xi_R(t_0)
+ b_p(T,t_0),
\label{eq:propagated_state}
\end{equation}
and enforcing $\xi_p(T) = 0$ gives the departure impulse
\begin{equation}
\Delta\xi_v(t_0) = -\,\xi_v(t_0)
- \tilde\Phi_{pv}^{-1}\bigl[\tilde\Phi_{pp}\,\xi_p(t_0)
+ \tilde\Phi_{pR}\,\xi_R(t_0) + b_p(T,t_0)\bigr].
\label{eq:departure_impulse}
\end{equation}
The residual velocity error immediately before arrival is then independent of
$\xi_v(t_0)$,
\begin{equation}
    \xi_v(T^-) = K_p\,\xi_p(t_0) + K_R\,\xi_R(t_0) + k_b ,
    \label{eq:arrival_velocity}
\end{equation}
with
\begin{equation}
    K_p := \tilde\Phi_{vp} - \tilde\Phi_{vv}\tilde\Phi_{pv}^{-1}\tilde\Phi_{pp},
    \qquad
    K_R := \tilde\Phi_{vR} - \tilde\Phi_{vv}\tilde\Phi_{pv}^{-1}\tilde\Phi_{pR},
    \qquad
    k_b := b_v - \tilde\Phi_{vv}\tilde\Phi_{pv}^{-1}b_p ,
    \label{eq:arrival_gains}
\end{equation}
and the arrival impulse is
\begin{equation}
\Delta\xi_v(T) = -\,\xi_v(T^-).
\label{eq:arrival_impulse}
\end{equation}
The physical increments follow from Lemma~\ref{lem:impulse_action},
\begin{equation}
\delta v_0 = \bar{R}(t_0)\,J_\ell^{\SO}\bigl(\xi_R(t_0)\bigr)\,\Delta\xi_v(t_0),
\qquad
\delta v_T = \bar{R}(T)\,J_\ell^{\SO}\bigl(\xi_R(T)\bigr)\,\Delta\xi_v(T),
\label{eq:physical_dv}
\end{equation}
with $\xi_R(T)$ given in closed form by~\eqref{eq:attitude_autonomous}.
\end{proposition}

\begin{proof}
Equation~\eqref{eq:propagated_state} is the position row of the forced response
$\xi(T) = \tilde\Phi\,\xi(t_0^+) + b$ with $\xi(t_0^+)$ given by
Lemma~\ref{lem:impulse_action}, which leaves $\xi_p(t_0)$ and $\xi_R(t_0)$
unaltered and shifts $\xi_v(t_0)$ by $\Delta\xi_v(t_0)$. Setting $\xi_p(T)=0$
and solving for $\Delta\xi_v(t_0)$ requires exactly the invertibility of
$\tilde\Phi_{pv}$, and yields~\eqref{eq:departure_impulse} after cancelling
$\tilde\Phi_{pv}^{-1}\tilde\Phi_{pv}\xi_v(t_0) = \xi_v(t_0)$; uniqueness follows
since the map $\Delta\xi_v(t_0)\mapsto\xi_p(T)$ is affine with invertible linear
part. For~\eqref{eq:arrival_velocity}, the velocity row of the same response is
$\xi_v(T^-) = \tilde\Phi_{vp}\xi_p(t_0) + \tilde\Phi_{vv}[\xi_v(t_0)+\Delta\xi_v(t_0)]
+ \tilde\Phi_{vR}\xi_R(t_0) + b_v$; substituting the bracket from
\eqref{eq:departure_impulse}, which equals
$-\tilde\Phi_{pv}^{-1}[\tilde\Phi_{pp}\xi_p(t_0)+\tilde\Phi_{pR}\xi_R(t_0)+b_p]$,
removes $\xi_v(t_0)$ and collects the coefficients~\eqref{eq:arrival_gains}.
Equation~\eqref{eq:arrival_impulse} is~\eqref{eq:rendezvous_condition} applied at
$T$, using continuity of $\xi_p$ and $\xi_R$ across the second impulse so that
$\xi_p(T^+) = 0$ is preserved.
\end{proof}

Equation~\eqref{eq:departure_impulse} decomposes into three physically distinct
actions: cancel the current velocity error, correct for the initial separation
through $\tilde\Phi_{pp}$, and correct for the attitude error through
$\tilde\Phi_{pR}$. The third term is the contribution absent from all classical
formulations.

\begin{remark}[Feasibility and degenerate transfer times]
\label{rem:feasibility}
Invertibility of $\tilde\Phi_{pv}(T,t_0)$ is the sole feasibility requirement. For the coasting circular case
the singular set is available in closed form: with $\tilde\Phi_{pv}$ reducing to
the classical $\Phi_{rv}$ of Corollary~\ref{cor:hcw_planner}, the out-of-plane
block $\sin(nT)/n$ vanishes at $nT = k\pi$ and the in-plane determinant
$n^{-2}\bigl[8 - 8\cos nT - 3nT\sin nT\bigr]$ vanishes at $nT = 2k\pi$, so within the first orbit the
degenerate horizons are $nT = k\pi$, $k \in \mathbb{Z}^+$---the classical
$\pi$-transfers. For thrusting or eccentric references these roots perturb away
from the closed-form values and must be located numerically. Beyond feasibility,
$\mathrm{cond}(\tilde\Phi_{pv})$ is a quantitative conditioning measure rather
than a binary test: by~\eqref{eq:departure_impulse}
and~\eqref{eq:impulse_norm_bounds}, $\|\delta v_0\|$ grows in proportion to
$\|\tilde\Phi_{pv}^{-1}\|$, so horizons near a degenerate value incur a
propellant penalty well before the solve itself becomes ill-posed. This motivates the quarter-orbit horizons
used in Sec.~\ref{sec:results}.
\end{remark}

\subsection{Special Cases}
\label{sec:dv_mismatch}

The planner of Proposition~\ref{prop:two_impulse} specializes along two
independent axes---whether the control is matched, and whether the reference
thrusts---with HCW at the intersection.

\begin{corollary}[Matched control]
\label{cor:matched}
If $\tilde\nu \equiv 0$ then $\tilde A = A$ by
Proposition~\ref{prop:forced_LTV}, hence $\tilde\Phi = \Phi$, the STM
of~\eqref{eq:LTV}, and $b \equiv 0$ by~\eqref{eq:forcing_integral}. The
departure impulse~\eqref{eq:departure_impulse} reduces to
\begin{equation}
\Delta\xi_v(t_0) = -\,\xi_v(t_0)
- \Phi_{pv}^{-1}\bigl[\Phi_{pp}\,\xi_p(t_0) + \Phi_{pR}\,\xi_R(t_0)\bigr],
\label{eq:departure_impulse_matched}
\end{equation}
and the residuals $\varepsilon$ of Theorem~\ref{thm:control_mismatch} vanish
identically, so the only propagation error is the gravity remainder bounded in
Theorem~\ref{thm:gravity_linearization}. Moreover $\|\xi_R\|$ is constant along
the transfer, since $\tfrac{d}{dt}\|\xi_R\|^2 = 2\,\xi_R^\top(-\bar\omega\times\xi_R) = 0$
by~\eqref{eq:xi_R_dot}; the bound coefficients $\alpha(\theta)$, $\beta(\theta)$
and the $\theta$-dependence of~\eqref{eq:bound_att} are therefore fixed by
$\xi_R(t_0)$ and hold uniformly on $[t_0,T]$ rather than requiring monitoring.
\end{corollary}

\begin{corollary}[Coasting reference]
\label{cor:coasting}
If $a_{\mathrm m} \equiv 0$ then $\tilde\Phi_{pR} = \tilde\Phi_{vR} = 0$ and
$K_R = 0$: the attitude error decouples from translation entirely, and both
log-coordinate impulses become independent of $\xi_R(t_0)$. Note that
$a_{\mathrm m} = \tfrac12(a + \bar a)$, so this occurs when neither vehicle thrusts during the transfer.
\end{corollary}

\begin{corollary}[HCW two-impulse planner]
\label{cor:hcw_planner}
Under the assumptions~\eqref{eq:hcw_assumptions} of Theorem~\ref{thm:hcw}, the
planner reduces to the classical HCW two-impulse solution. Writing
$s_\Delta = \sin n\Delta t$, $c_\Delta = \cos n\Delta t$, the relevant block is
\begin{equation}
    \Phi_{pv} = \Phi_{rv} = \frac{1}{n}
    \begin{bmatrix}
        s_\Delta & 2(1-c_\Delta) & 0 \\
        -2(1-c_\Delta) & 4s_\Delta - 3n\Delta t & 0 \\
        0 & 0 & s_\Delta
    \end{bmatrix},
    \label{eq:hcw_Phi_rv}
\end{equation}
the classical position--velocity partition of the HCW state transition matrix,
and~\eqref{eq:departure_impulse_matched} becomes
\begin{equation}
    \Delta\xi_v(t_0) = -\,\dot\xi_p(t_0) - \Phi_{rv}^{-1}\Phi_{rr}\,\xi_p(t_0),
    \label{eq:hcw_departure}
\end{equation}
the textbook two-impulse formula in LVLH coordinates.
\end{corollary}

\begin{proof}
By Corollaries~\ref{cor:matched} and~\ref{cor:coasting} the planner uses $\Phi$
with $\Phi_{pR} = 0$. Theorem~\ref{thm:hcw} identifies $\xi_p$ with the LVLH
relative position and gives $\xi_v = \dot\xi_p + \Omega\times\xi_p$
by~\eqref{eq:inv_hcw_p}, so the two coordinate sets are related by the constant
similarity $T_c = \left[\begin{smallmatrix} I & 0\\ \skewmat{\Omega} & I\end{smallmatrix}\right]$
and $\Phi = T_c\,\Phi^{\mathrm{HCW}}T_c^{-1}$, where $\Phi^{\mathrm{HCW}}$ has
blocks $\Phi_{rr},\Phi_{rv},\Phi_{vr},\Phi_{vv}$ in the classical
$(\xi_p,\dot\xi_p)$ ordering. Since the first block row of $T_c$ is
$[\,I\;\;0\,]$ and the second block column of $T_c^{-1}$ is $[\,0\;\;I\,]^\top$,
the $pv$ block is unchanged, $\Phi_{pv} = \Phi_{rv}$, giving
\eqref{eq:hcw_Phi_rv}, while
$\Phi_{pp} = \Phi_{rr} - \Phi_{rv}\skewmat{\Omega}$. Substituting into
\eqref{eq:departure_impulse_matched},
\begin{equation*}
    \Delta\xi_v(t_0)
    = -\xi_v(t_0) - \bigl[\Phi_{rv}^{-1}\Phi_{rr} - \skewmat{\Omega}\bigr]\xi_p(t_0)
    = -\bigl[\xi_v(t_0) - \Omega\times\xi_p(t_0)\bigr] - \Phi_{rv}^{-1}\Phi_{rr}\,\xi_p(t_0),
\end{equation*}
and the bracket is $\dot\xi_p(t_0)$ by~\eqref{eq:inv_hcw_p}, giving
\eqref{eq:hcw_departure}. Finally $\xi_R \equiv 0$ gives
$J_\ell^{\SO} = I$, so~\eqref{eq:physical_dv} reduces to
$\delta v_0 = \bar R(t_0)\Delta\xi_v(t_0)$ and the log-coordinate impulse is the
LVLH-resolved physical increment.
\end{proof}

\begin{corollary}[Thrusting chief, coasting deputy]
\label{cor:thrusting_chief}
If the chief executes $\bar a(t) \neq 0$ while the deputy coasts between
impulses, with angular velocity matching the chief's, then $\tilde\nu(t) = (0,\,-\bar a(t),\,0)^\top$ and
$a_{\mathrm m} = \tfrac12\bar a$ are both determined by the reference trajectory
alone, so the planner requires no information beyond it. Only $b_v$ is directly
forced in~\eqref{eq:joint_integration}; $b_p$ accumulates from $b_v$ through
$\tilde\Phi_{pv}$ over the transfer and is the term distinguishing
\eqref{eq:departure_impulse} from its matched-control
form~\eqref{eq:departure_impulse_matched} in this scenario.
\end{corollary}

\subsection{Remarks on Scope and Extension}
\label{sec:dv_remarks}

\begin{remark}[Attitude--translation coupling]
\label{rem:coupling}
The $\tilde\Phi_{pR}\xi_R(t_0)$ term in~\eqref{eq:departure_impulse} is the
structural difference between this planner and every formulation derived from an
unforced reference. It is absent whenever $a_{\mathrm m} = 0$
(Corollary~\ref{cor:coasting}), which is why classical planners have never
needed it, and in position it grows as $\|a_{\mathrm m}\|\,\theta\,\Delta t^2$ rather than with
the tidal magnitude $n^2\|\xi_p\|$.

Comparing the two accelerations entering the velocity row
of~\eqref{eq:Atilde_block} gives a criterion, evaluable before any planning, for
whether the additional machinery of this framework is warranted over a classical
planner:
\begin{equation}
    \kappa := \frac{\|a_{\mathrm m}\|\,\theta}{n^2\,\|\xi_p\|} .
    \label{eq:coupling_ratio}
\end{equation}

Since $\lVert a_{\mathrm m} \times \xi_R \rVert \le \lVert a_{\mathrm m} \rVert\,\theta$
and $\lVert G(t)\,\xi_p \rVert \ge n^2 \lVert \xi_p \rVert$, $\kappa$ bounds the ratio
of the two accelerations from above; the coupling vanishes when $\xi_R$ is
parallel to $a_{\mathrm m}$.
For $\kappa \ll 1$ the coupling is a correction to the tidal term and an
HCW-class planner suffices; for $\kappa \gtrsim 1$ the coupling can be the dominant
term in the translational error dynamics and no planner built on an unforced
reference can capture it, at any separation. Both quantities are known in
advance: $a_{\mathrm m}$ from the maneuver plan, $\theta$ and $\|\xi_p\|$ from
navigation. The ratio $\|a_{\mathrm m}\|/n^2$ is a property of the vehicle and
orbit alone, and it is what sets the scale. For a chemical thruster
($10^{-2}$~m/s$^2$) on a $500$~km circular orbit it is $8.1\times10^{3}$~m, so at
a separation of $114$~m the coupling overtakes the tidal term at
$\theta \approx 0.8^\circ$. In geosynchronous orbit the mean motion is smaller by
a factor of fifteen and the ratio grows by more than two orders of magnitude:
even for low-thrust electric propulsion ($10^{-4}$~m/s$^2$) it is
$1.9\times10^{4}$~m, placing the crossover at $\theta \approx 0.35^\circ$.
Geosynchronous station-keeping with a nearby deputy is therefore the regime in
which classical planners fail most severely. Conversely, low-thrust propulsion in
low Earth orbit gives $\|a_{\mathrm m}\|/n^2 \approx 81$~m, for which
$\kappa < 1$ throughout the range of attitude errors where the linearization of
Remark~\ref{rem:validity_control} is valid. There the coupling is real but
subdominant, and correcting it gains little, because the planner's own residual,
the attitude-induced gravity term $\delta_{\mathrm{att}}$ of
Theorem~\ref{thm:gravity_linearization}, is then of the same size as the coupling it
removes; Section~\ref{sec:num_mechanism} shows a classical planner can be the more
accurate of the two in that regime. Where $\kappa \gtrsim 1$, neglecting the coupling
produces a terminal position error proportional to $\xi_R(t_0)$.

The criterion classifies the regime; it does not by itself predict the size of
the benefit, which Section~\ref{sec:num_attitude} shows follows $\kappa$ only where
the proposed planner's residual is linear in $\theta$.
\end{remark}

\begin{remark}[Relation to a Cartesian formulation]
\label{rem:cartesian}
Because the attitude histories of both vehicles are known in the cooperative
setting assumed here, the thrust--attitude forcing can also be carried in an
inertial Cartesian linearization of the relative motion,
$\ddot\rho = G_I(t)\,\rho + \bigl(R(t) - \bar R(t)\bigr)a$. Retained exactly, that
forcing leaves only $\delta_{\mathrm{grav}}$, but its particular integral depends
on the attitude error and must be recomputed whenever the attitude estimate
changes. A Cartesian planner with the precompute-once structure of
Proposition~\ref{prop:two_impulse} must instead linearize the forcing,
$(R - \bar R)a \approx -\bar R\skewmat{a}\xi_R$, and so incurs a residual
acceleration of order $\tfrac12\theta^2\|a\|$. In log coordinates and under matched control the coupling
$-\skewmat{a_{\mathrm m}}\xi_R$ of~\eqref{eq:A_matrix} is exact, since it is part of
the linear term $-\ad_{\bar\nu}\,\xi$ of~\eqref{eq:loglinear}, and the price is
paid instead in $\delta_{\mathrm{att}}$, of order $\theta\,n^2\|\xi_p\|$
by~\eqref{eq:bound_att}. The ratio of the two residuals is therefore of order
$\kappa$: the log-coordinate planner is the more accurate of the two
attitude-independent planners in the thrust-dominated regime that motivates
this work, and the Cartesian one when $\kappa$ is of order one or smaller.
This comparison is for matched control. When the deputy coasts while the chief
thrusts, the Cartesian forcing reduces to $-\bar R(t)\bar a$, which does not
depend on the relative attitude, so a Cartesian planner carries it exactly with the
same precompute-once structure.
Section~\ref{sec:num_mechanism} places the crossover between $\kappa = 1.4$ and
$1.8$ on the three scenarios that reach it.

The term $\delta_{\mathrm{att}}$ arises because the tidal tensor acts on
$\rho = \bar R J_\ell^{\SO}(\xi_R)\,\xi_p$ rather than on $\xi_p$ itself. Replacing
$G$ in~\eqref{eq:A_matrix} by $J_\ell^{-1,\SO}(\xi_R(t))\,G\,J_\ell^{\SO}(\xi_R(t))$,
which is known in advance by~\eqref{eq:attitude_autonomous}, removes it, at the
cost of a transition matrix that depends on $\xi_R(t_0)$ and must be recomputed
with it. Approximating the conjugated tensor by $G$ is thus a design choice: it is
what makes $\tilde\Phi$ independent of the attitude error, so that it is computed
once per reference arc and reused across attitude estimates and replans, and what
exposes the coupling as the explicit block $\tilde\Phi_{pR}$ and the
criterion~\eqref{eq:coupling_ratio}. Second-order state transition
tensors~\cite{park2006} would capture the remaining quadratic terms of either
formulation while preserving the precompute-once structure, at the cost of a
$9\times9\times9$ tensor and a nonlinear impulse solve.
\end{remark}

\begin{remark}[Impulsive approximation]
\label{rem:impulsive}
Lemma~\ref{lem:impulse_action} models each maneuver as instantaneous. The
requirement is the usual one---burn duration short relative to the transfer
horizon---with the additional condition, specific to thrusting references, that
$\bar a$ and $\bar\omega$ be effectively constant over the burn, so that
$\bar R$ and $\xi_R$ in~\eqref{eq:physical_dv} are well defined at a single
instant. Finite-burn corrections may be applied by the standard practice of
centering the burn arc on $t_0$; Section~\ref{sec:num_envelope} quantifies both
the admissible burn duration and the rotation mechanism identified here.
\end{remark}

\begin{remark}[Nonzero waypoints]
\label{rem:waypoint}
Targeting a standoff point $\rho_{\mathrm{target}} \neq 0$, as in inspection or
station-keeping, replaces $\xi_p(T) = 0$ by
$\xi_p(T) = \xi_p^{\mathrm{target}}$ with
\begin{equation}
    \xi_p^{\mathrm{target}}
    = J_\ell^{-1,\SO}\bigl(\xi_R(T)\bigr)\,\bar R(T)^\top\rho_{\mathrm{target}} ,
    \label{eq:waypoint_target}
\end{equation}
by~\eqref{eq:rho_xi_relation}. This is again exact, and computable in advance,
precisely because the attitude row is autonomous: $\xi_R(T)$ is available from
\eqref{eq:attitude_autonomous} before any impulse is chosen. Equation
\eqref{eq:departure_impulse} then acquires the additional term
$+\,\tilde\Phi_{pv}^{-1}\xi_p^{\mathrm{target}}$, and nothing else in the
derivation changes.
\end{remark}

%

\section{Numerical Results}\label{sec:results}

The experiments in this section validate the three theorems, the planner of
Section~\ref{sec:dv_planning}, and the conditions under which the framework is
intended to be used. Section~\ref{sec:num_setup} states the scenarios.
Sections~\ref{sec:num_hcw}--\ref{sec:num_gravity} follow the development in
order: exact recovery of the classical equations, isolation of each physical
mechanism and computational cost, the attitude scaling and the criterion of
Remark~\ref{rem:coupling}, mismatched control, and the gravity bounds. Section~\ref{sec:num_envelope}
closes with the applicability envelope, the effect of replanning, and the
perturbations the model omits. Every terminal error reported is measured on the
nonlinear truth model after the planner's impulses have been applied, and the
comparison throughout includes planners that are better than the proposed one in
some regimes.

\subsection{Setup}\label{sec:num_setup}

All cases share a single nominal initial relative state,
\begin{equation}
    \xi_p(0) = [100,\,50,\,20]^\top~\mathrm{m},
    \qquad
    \xi_v(0) = [0.1,\,-0.05,\,0]^\top~\mathrm{m/s},
    \qquad
    \|\xi_R(0)\| = 5^\circ \ \text{about} \ \tfrac{1}{\sqrt3}[1,1,1]^\top .
    \label{eq:nominal_state}
\end{equation}
The separation $\|\xi_p(0)\| = 113.6$~m and closing rate
$\|\xi_v(0)\| = 0.112$~m/s place the transfer in the close-range phase of a
rendezvous, inside the 500~km proximity-operations regime
of~\cite{petersen2024} and at the scale at which a two-impulse plan is executed
rather than an approach corridor negotiated. The attitude error is the relative
orientation between the two body frames, not the pointing error of either
vehicle: a deputy holding its docking port or approach sensor on the chief while
the chief holds a local-vertical local-horizontal attitude is separated from it
by a finite rotation regardless of how well each vehicle knows its own attitude,
and $5^\circ$ is a modest value for that separation.

No single initial condition is representative, and each choice
in~\eqref{eq:nominal_state} that materially affects the results is accompanied
below by a sweep or an ensemble over that choice. The attitude magnitude is
swept from $0$ to $20^\circ$ in Section~\ref{sec:num_attitude}; the attitude
axis over a $200$-point lattice on the sphere and the thrust direction over a
$100$-point lattice in \ref{app:axis}; the separation from $113.6$~m to
$1000$~km in Section~\ref{sec:num_gravity}; and the separation, velocity, and
attitude jointly in a 500-sample ensemble per scenario ($\|\xi_p\|$ uniform on
$[10, 500]$~m, $\|\xi_v\|$ on $[0, 0.5]$~m/s, $\|\xi_R\|$ on $[0, 10^\circ]$,
uniform directions and axes). In that ensemble the nominal state falls at the
thirty-third percentile of terminal error on S1 and the forty-fourth on S3, and
the Theorem~\ref{thm:gravity_linearization} bounds hold on all $1000$
trajectories. The magnitudes are drawn uniformly in norm rather than in volume,
which weights small separations; half of the S1 draws have $\theta < 5^\circ$.
On $200$ further S1 draws the translational-only planner defined below is less
accurate than the proposed planner in every case, by a factor of at least $7.5$
and a median of $23$. Navigation and burn-duration perturbations are treated in
Section~\ref{sec:num_envelope}. The nominal attitude axis is a favorable draw
(\ref{app:axis}), so Table~\ref{tab:mechanism} carries a median-axis row
alongside it.

Body-frame thrust is directed along $\hat b_1$, radial at $t_0$ in the
local-vertical local-horizontal frame, except in the mismatched-control cases of
Section~\ref{sec:num_mismatch}, where it is along-track; \ref{app:axis}
shows that this choice changes the terminal error by up to a factor of fifteen.
Nonlinear propagation uses full two-body Newtonian gravity with integration
tolerances $\mathrm{rtol} = 10^{-12}$, $\mathrm{atol} = 10^{-14}$, and the truth
model agrees with an independent Modelica implementation to $1.2\times10^{-10}$
relative. Portions of the simulation code were developed with the assistance of a
generative AI tool, and the results were verified against an independent
implementation of the planner and truth model written from the equations of this
paper. Deputy states are initialized from the reference through the $\SE$
exponential map, so that all comparators begin from the same physical state and
differ only in the applied $\Delta v$.

Five planners are compared. The \emph{proposed} planner integrates the full
$9\times9$ system matrix, including the coupling blocks. The
\emph{translational-only} planner integrates the $6\times6$ translational block
of~\eqref{eq:A_matrix} along the same reference arc and omits
$\tilde\Phi_{pR}$; it is the proposed planner with the coupling removed, and so
isolates its effect. The classical \emph{Hill--Clohessy--Wiltshire} planner uses
the closed form with the chief's mean motion, including on elliptical references,
as its derivation requires. On elliptical references the
\emph{Yamanaka--Ankersen} state transition matrix~\cite{yamanaka2002} separates
the cost of assuming a circular orbit from the value of the present framework; at
zero eccentricity it coincides with the classical planner. These three planners
take the log-coordinate state as their initial condition and apply the departure
impulse through the map of Lemma~\ref{lem:impulse_action}, so that each differs
from the proposed planner only in its transition matrix.
The fifth, \emph{Cartesian (linearized)}, is the
comparator of Remark~\ref{rem:cartesian}: it integrates the inertial
relative motion
\begin{equation}
    \ddot\rho = G_I(t)\,\rho + F(t)\,\xi_R(t_0),
    \qquad
    F(t) = -\bar R(t)\skewmat{\bar a}\,\Exp\bigl(-\skewmat{\bar\omega}(t-t_0)\bigr),
    \label{eq:cartesian_lin}
\end{equation}
in which the thrust forcing is linearized in the attitude error, and precomputes a
$6\times6$ transition matrix and a $6\times3$ attitude sensitivity along the
reference, neither of which depends on $\xi_R(t_0)$. It therefore has the same
precompute-once structure as the proposed planner and is the fair Cartesian
counterpart to it.

Table~\ref{tab:scenarios} lists the reference scenarios. The final column gives
the criterion $\kappa$ of~\eqref{eq:coupling_ratio} evaluated at the initial
state, so that the regime of each scenario is visible before any result is
presented. The scenarios span nearly three orders of magnitude in $\kappa$ and
bracket the crossover $\kappa = 1$.

\begin{table}[t]
\caption{Reference scenarios. $\kappa$ is evaluated at
         $\|\xi_p(t_0)\| = 113.6$~m and $\theta(t_0) = 5^\circ$; on the
         Molniya rows $n$ is the mean motion.}
\label{tab:scenarios}
\centering
\begin{tabular}{cllcccc}
\hline\hline
ID & Orbit & $\|a_{\mathrm m}\|$, m/s$^2$ & $n$, rad/s
   & $\|a_{\mathrm m}\|/n^2$, m & $T$ & $\kappa$ \\
\hline
S1 & LEO 500 km circular & $10^{-2}$ & $1.109\times10^{-3}$ & $8.1\times10^{3}$ & 23.6 min & $6.25$ \\
S2 & LEO 500 km circular & $10^{-4}$ & $1.109\times10^{-3}$ & $81$              & 23.6 min & $0.063$ \\
S3 & GEO circular        & $10^{-4}$ & $7.292\times10^{-5}$ & $1.9\times10^{4}$ & 90 min   & $14.4$ \\
S4 & GEO circular        & $10^{-5}$ & $7.292\times10^{-5}$ & $1.9\times10^{3}$ & 90 min   & $1.44$ \\
S5 & Molniya, $e = 0.74$ & $10^{-3}$ & $1.459\times10^{-4}$ & $4.7\times10^{4}$ & 179.5 min & $36.1$ \\
S6 & LEO 500 km circular & $0$       & $1.109\times10^{-3}$ & $0$               & 23.6 min & $0$ \\
S7 & Molniya, $e = 0.74$ & $0$       & $1.459\times10^{-4}$ & $0$               & 179.5 min & $0$ \\
\hline\hline
\end{tabular}
\end{table}

Circular cases use quarter-orbit horizons to avoid the $\pi$-transfer
singularity~\cite{prussing1969}. The geosynchronous horizon is shortened to
90~min ($nT \approx 0.39$), which is better conditioned than the six-hour quarter
orbit, $\mathrm{cond}(\tilde\Phi_{pv})$ of $1.1$ against $2.9$, and $2\%$ cheaper.

\subsection{Exact Recovery of the Classical Equations}\label{sec:num_hcw}

Under the assumptions~\eqref{eq:hcw_assumptions} of Theorem~\ref{thm:hcw} the
system matrix $A(t)$ of~\eqref{eq:A_matrix} is constant: $\bar\omega = \Omega$,
$\bar a \equiv 0$, and the body-frame tidal tensor~\eqref{eq:G_evaluated} is
fixed in the local-vertical local-horizontal frame. The state transition matrix
is therefore $\Phi = \Exp(At)$, and Theorem~\ref{thm:hcw} is an algebraic
identity between $\Exp(At)$ and the classical closed form rather than a
statement about convergence.

Evaluated blockwise against the textbook Hill--Clohessy--Wiltshire transition
matrix~\cite{schaub_junkins} under the similarity $T_c$ of
Corollary~\ref{cor:hcw_planner}, the relative error ranges from $1.5\times10^{-16}$ to $1.3\times10^{-15}$ over $nT \in (0, \pi]$ and remains below $8.8\times10^{-15}$ through $nT = 2\pi$, which is machine precision. The planner identity of Corollary~\ref{cor:hcw_planner} is recovered at
the same level: the departure impulse~\eqref{eq:departure_impulse_matched} and
the textbook two-impulse formula~\eqref{eq:hcw_departure} agree to $3\times10^{-15}$ relative. The invariance of the zero-attitude-error set used in the
proof of Theorem~\ref{thm:hcw}, and the constancy of $\|\xi_R\|$ asserted in
Corollary~\ref{cor:matched}, hold to the same precision: across all eight trajectories of
Table~\ref{tab:gravity}, $\|\xi_R(t)\|$ departs from its initial value by no
more than $4\times10^{-13}$~rad.

Propagating instead from a numerically integrated reference introduces an error
floor of $8.7\times10^{-13}$ that tracks the reference tolerance rather than the
identity, which quantifies how accurately the reference arc must be known before
state transition matrix error becomes visible.

\subsection{Mechanism Isolation}\label{sec:num_mechanism}

Table~\ref{tab:mechanism} reports terminal position error
$\|\bar p(T) - p(T)\|$ from nonlinear ground truth. Rows~A through~F each add
exactly one physical effect relative to the row above; the middle group places
the same measurement at three further values of $\kappa$.

\begin{table}[t]
\caption{Terminal position error from nonlinear ground truth, m. The
         Yamanaka--Ankersen planner coincides with the classical one on circular
         rows and is reported on elliptical rows only. On the elliptical rows,
         differences below about $3$~mm are at the resolution of the truth
         model.}
\label{tab:mechanism}
\centering
\small
\setlength{\tabcolsep}{4pt}
\begin{tabular}{cllcrrrrr}
\hline\hline
Row & Case & $\|\xi_R\|$ & $\kappa$ & Proposed & Trans.-only & YA & HCW & Cartesian (lin.) \\
\hline
A & S6, circular coasting   & $0^\circ$ & $0$     & $0.0022$ & $0.0022$ & ---      & $0.0022$ & $0.0022$ \\
B & \quad $+$ thrust (S1)   & $0^\circ$ & $0$     & $0.0022$ & $0.0022$ & ---      & $0.0505$ & $0.0022$ \\
C & \quad $+$ attitude (S1) & $5^\circ$ & $6.25$  & $9.02$   & $486.1$  & ---      & $486.0$  & $21.3$ \\
  & \quad\quad median axis  & $5^\circ$ & $6.25$  & $23.2$   & $786.2$  & ---      & $786.2$  & $47.0$ \\
\hline
  & S2, low thrust          & $5^\circ$ & $0.063$ & $5.62$   & $1.71$   & ---      & $1.72$   & $0.215$ \\
  & S4, near the crossover  & $5^\circ$ & $1.44$  & $0.529$  & $9.00$   & ---      & $9.00$   & $0.424$ \\
  & S3, GEO station-keeping & $5^\circ$ & $14.4$  & $0.523$  & $94.7$   & ---      & $94.7$   & $4.24$ \\
\hline
D & S7, $+$ eccentricity    & $0^\circ$ & $0$     & $0.0174$ & $0.0174$ & $0.0178$ & $2764$   & $0.0178$ \\
E & \quad $+$ thrust (S5)   & $0^\circ$ & $0$     & $0.0182$ & $0.0179$ & $0.101$  & $2766$   & $0.0182$ \\
F & \quad $+$ attitude (S5) & $5^\circ$ & $36.1$  & $93.9$   & $4456$   & $4456$   & $7177$   & $124.8$ \\
\hline\hline
\end{tabular}
\end{table}

Five observations follow.

First, thrust alone cannot degrade planning accuracy at zero attitude error.
With $\tilde\nu \equiv 0$ and $\xi_R \equiv 0$ the coupling block
$-\skewmat{a_{\mathrm m}}$ of~\eqref{eq:Atilde_block} multiplies zero, so
body-frame thrust reaches the translational dynamics through no channel except
the reference ephemeris entering $G(t)$. This is the transition from Row~A to
Row~B, in which every planner integrated along the true reference is unchanged at
$2.2$~mm and the classical planner accumulates $50.5$~mm from the variation of
$G(t)$. The same mechanism explains Row~E. The thrusting reference departs
$46.6$~km from the Keplerian arc that the Yamanaka--Ankersen closed form presumes
(for comparison, $\tfrac12\|\bar a\|T^2 = 58$~km), and integrating the relative
dynamics on the coasting ephemeris predicts a $0.111$~m miss, against the
observed $0.101$~m. The Yamanaka--Ankersen error in Row~E is therefore ephemeris
mismatch, to within about $10\%$, rather than a failure of its relative dynamics.

Second, once the attitude error is nonzero the coupling term is the only channel
by which thrust reaches translation, and in Row~C it dominates. The
translational-only and classical planners agree to within $0.1$~m there, since
the $50.5$~mm attributable to the variation of $G(t)$ in Row~B is negligible at
this scale; the coupling is the sole material difference between them and the
proposed planner, which removes $98.1\%$ of their terminal error. Row~F is the
clearest case. The Yamanaka--Ankersen planner, which treats arbitrary
eccentricity exactly under coasting and agrees with the proposed planner to
within $0.5$~mm in Row~D, fails at $4456$~m once attitude error meets thrust,
identically to the translational-only planner (the two differ by two parts in
$10^5$), because neither contains a coupling channel. The proposed planner
reaches $93.9$~m. The gap is therefore the coupling itself, which no formulation
built on an unforced reference contains.

Third, the terminal velocity condition is met alongside the position condition,
since both parts of~\eqref{eq:rendezvous_condition} are enforced by the planner.
In Row~C the terminal velocity error $\|v(T^+) - \bar v(T)\|$ is
$1.17\times10^{-2}$~m/s against $0.779$~m/s for the classical planner, a factor
of $67$; on S3 it is $1.40\times10^{-4}$ against $3.29\times10^{-2}$~m/s, a
factor of $235$; and in Row~F it is $1.38\times10^{-2}$ against $1.27$~m/s for the
classical planner and $0.887$~m/s for the translational-only planner, factors of
$92$ and $64$. At zero attitude error the proposed and translational-only planners
meet the velocity condition to about $2\times10^{-6}$~m/s, while the
classical planner reaches $1.6\times10^{-4}$~m/s on the thrusting circular
reference and $0.38$~m/s on the elliptical references. The coupling correction is
not free: in Row~C the proposed plan costs $0.855$~m/s against $0.313$~m/s for
the classical one, three quarters of the difference at arrival.

Fourth, the proposed planner's advantage over the coupling-blind planners
follows the criterion of Remark~\ref{rem:coupling}. The two geosynchronous rows
differ by a factor of ten in $\kappa$ and by a factor of $10.6$ in the ratio of
classical to proposed error, from $17.0$ at $\kappa = 1.44$ to $181$ at
$\kappa = 14.4$; in low Earth orbit the ratio is $54$ at $\kappa = 6.25$. At
$\kappa = 0.063$ the order reverses: on S2 the classical and translational-only
planners reach $1.7$~m against $5.6$~m for the proposed planner. At low $\kappa$
there is little coupling to correct, and the proposed planner's error is almost
entirely its own residual $\delta_{\mathrm{att}}$, whose propagated effect is
$6.5$~m, against $2$~mm from $\delta_{\mathrm{grav}}$. The coupling-blind error
is instead the uncorrected coupling, which is odd in the initial state, plus
second-order terms, which are even; in the flown geometry they partly cancel, and
for the mirror-image initial state $-\xi(t_0)$ the translational-only error is
$10.6$~m. At $\kappa \ll 1$ which of the two planners is more accurate is
therefore a matter of geometry, and the regime calls for neither: S2 is where the
criterion indicates that the coupling machinery is unnecessary.

Fifth, the Cartesian planner of~\eqref{eq:cartesian_lin} confirms the trade
described in Remark~\ref{rem:cartesian}. It is the most accurate planner on S2,
at $0.215$~m, and remains better than the proposed planner on S4, at $0.424$
against $0.529$~m. On S1, S3, and Row~F the order reverses, by factors of $2.4$,
$8.1$, and $1.3$. Sweeping the attitude error (Fig.~\ref{fig:attitude_scaling}a) places the
crossover at $\theta = 1.14^\circ$ on S1, $0.63^\circ$ on S3, and $6.23^\circ$ on S4,
corresponding to $\kappa = 1.42$, $1.81$, and $1.80$; on S2, where $\kappa$ does
not exceed $0.25$ over the swept range, the Cartesian planner is better at every
attitude error. On the geosynchronous scenarios, where the proposed residual is
linear in $\theta$ and the Cartesian one grows faster, the advantage grows with
$\kappa$, reaching a factor of $26$ on S3
at $20^\circ$; on S1 both residuals grow faster than $\theta$ and the factor is
$2.4$ at $5^\circ$ and $2.2$ at $20^\circ$. Retaining the forcing
exactly, rather than linearizing it, gives $1.69$~mm in Row~C, and so does the
proposed planner with $G$ replaced by its conjugate
$J_\ell^{-1,\SO}GJ_\ell^{\SO}$, so the entire gap between the proposed planner
and an exact one in that row is $\delta_{\mathrm{att}}$.

Retaining the left Jacobian in the impulse map of
Lemma~\ref{lem:impulse_action} is required for the same reason that the coupling
block is. The naive identification $\delta v = \bar R\,\Delta\xi_v$ rotates the
component of every impulse perpendicular to $\xi_R$ by exactly $\theta/2$, an
identity rather than a leading-order result, and running the planner with it
raises the Row~C terminal error from $9.02$ to $15.13$~m, an increase of $68\%$.

The proposed planner is also inexpensive and close to propellant-optimal. Accuracy
is limited by the linearization rather than by integration: five fixed-step
fourth-order Runge--Kutta steps over the Row~C transfer change the terminal error
by $2\%$. Since $\tilde\Phi_{Rp} = \tilde\Phi_{Rv} = 0$
by~\eqref{eq:attitude_autonomous}, the integration needs $72$ states and costs
$23.7$~kflop, against $308$~flops for the classical planner; with $\tilde\Phi$,
$b$, and the factorization of $\tilde\Phi_{pv}$ cached, re-solving for the
impulses costs $342$~flops. The resulting plan is within $0.1$ to $0.4\%$ of the
minimum-cost two-impulse transfer found by nonlinear single shooting on the true
dynamics.

Figure~\ref{fig:separation} shows the inertial separation and the prediction
residual for three cases of Table~\ref{tab:mechanism}.

\begin{figure}[t]
    \centering
    \includegraphics[width=\textwidth]{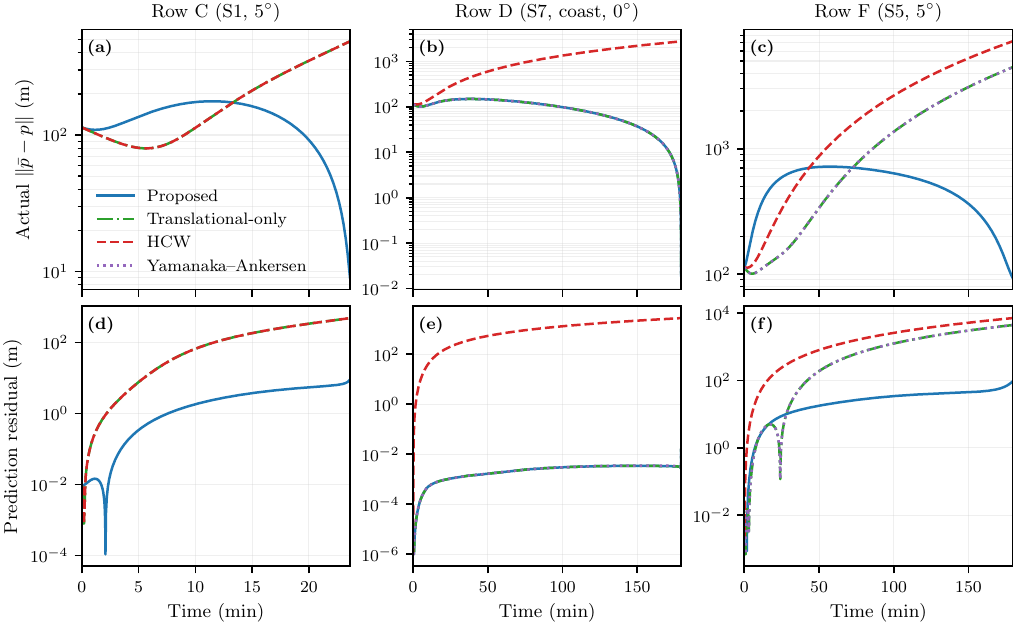}
    \caption{Inertial separation (top) and prediction residual (bottom) for
             (a,d) Row~C, (b,e) Row~D, and (c,f) Row~F of
             Table~\ref{tab:mechanism}.}
    \label{fig:separation}
\end{figure}

\subsection{Attitude Scaling and the Coupling Criterion}\label{sec:num_attitude}

Figure~\ref{fig:attitude_scaling} sweeps $\|\xi_R(0)\|$ from $0$ to $20^\circ$ at
fixed thrust. The coupling-blind errors grow nearly linearly in $\theta$, as the
coupling term $\|a_{\mathrm m}\|\,\theta\,T^2$ predicts. The proposed planner's
residual is also linear on S2, S3, and S4, and on S1 below $2^\circ$, but grows as
$\theta^{1.92}$ on S1 and $\theta^{1.42}$ on S5 at larger attitude errors. Both
behaviors come from $\delta_{\mathrm{att}}$, the only residual of first order in
$\theta$ under matched control, since the coupling term is retained exactly. Where
the attitude contribution to the departure impulse~\eqref{eq:departure_impulse} is
large, a larger attitude error produces a larger mid-course excursion, and
$\delta_{\mathrm{att}}$, being proportional to $\|\xi_p(t)\|$, grows with it.
Propagating $\delta_{\mathrm{att}}$ through $\tilde\Phi_{pv}$ reproduces the
terminal error to within $0.2\%$, while $\delta_{\mathrm{grav}}$ is smaller by
three to four orders of magnitude. The attitude-induced gravity term is therefore
the proposed planner's accuracy floor, and its scaling with $\theta$ determines
where the ratio to the classical error follows $\kappa$ and how the comparison with
the Cartesian planner evolves with attitude error
(Section~\ref{sec:num_mechanism}).

\begin{figure}[t]
    \centering
    \includegraphics[width=\textwidth]{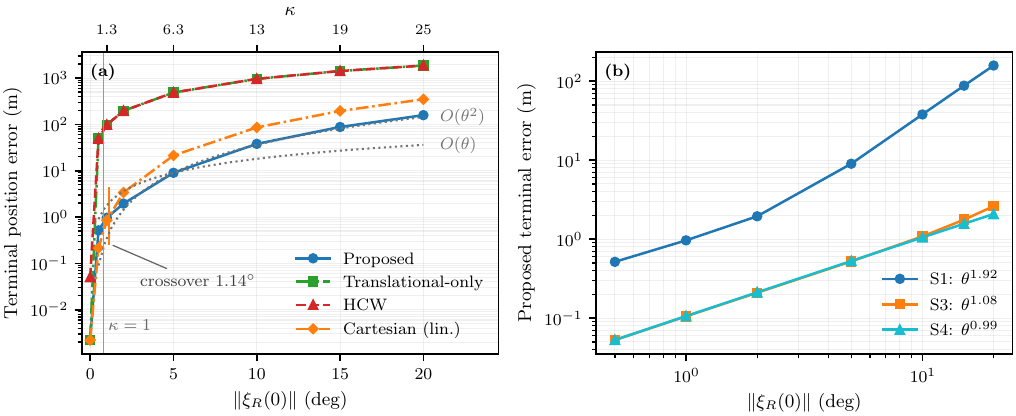}
    \caption{(a) Terminal error versus attitude-error magnitude on S1 for the
             proposed, translational-only, classical, and Cartesian (linearized)
             planners, with the corresponding $\kappa$ on the upper axis and the
             Cartesian crossover at $1.14^\circ$ marked. (b) The proposed
             planner's residual on S1, S3, and S4, with fitted exponents.}
    \label{fig:attitude_scaling}
\end{figure}

On S1, S3, and S4 the coupling correction removes more than $91\%$ of the
coupling-blind error at every tested attitude error. On S1 it removes $99\%$
below $2^\circ$, falling to $98.1\%$ at $5^\circ$, $96.0\%$ at $10^\circ$, and
$91.5\%$ at $20^\circ$, which is the excursion growth described above and not a
failure of the coupling linearization; on S3 it lies between $99.3\%$ and
$99.5\%$, and on S4 between $94.1\%$ and $94.2\%$, independently of $\theta$. On
S2 there is no net correction, for the reasons given in
Section~\ref{sec:num_mechanism}. The validated envelope therefore extends to
$20^\circ$ wherever $\kappa \gtrsim 1$, well beyond the few degrees anticipated in
Remark~\ref{rem:validity_control}.

By~\eqref{eq:coupling_ratio}, the crossover $\kappa = 1$ occurs at $0.80^\circ$
for the chemical-thrust low-Earth-orbit case, $0.35^\circ$ for the geosynchronous
electric-propulsion case, and $3.46^\circ$ on S4. Where the
residual is linear, both errors scale with $\theta$, the attitude dependence
cancels, and the ratio of classical to proposed error follows
$\kappa/\theta = \|a_{\mathrm m}\|/(n^2\|\xi_p\|)$: on S4 it lies between $17.0$
and $17.3$ at every attitude error, and over $200$ random S3 states it tracks
$\|a_{\mathrm m}\|/(n^2\|\xi_p\|)$ with a rank correlation of $0.88$. In low Earth
orbit at high thrust, where the residual is superlinear, the same ratio over
$200$ random S1 states has a median between $20$ and $26$ in every quintile of
$\kappa$ and no significant correlation with it.

\subsection{Mismatched Control}\label{sec:num_mismatch}

The experiments above use matched control, for which the residuals of
Theorem~\ref{thm:control_mismatch} vanish identically. Here the deputy's inputs
differ from the chief's in three ways: it coasts while the chief thrusts, so that
$\tilde\nu(t) = (0,-\bar a(t),0)^\top$ and $a_{\mathrm m} = \tfrac12\bar a$ are
fixed by the reference (Corollary~\ref{cor:thrusting_chief}); it slews $5^\circ$
about its body $x$ axis over $60$~s with matched acceleration; or its thrust is
misaligned by $10\%$. These cases validate Theorem~\ref{thm:control_mismatch} and
the forced planner; they are not a comparison with the Cartesian planner, which
for a coasting deputy carries the chief's thrust exactly
(Remark~\ref{rem:cartesian}). Table~\ref{tab:forced} compares the forced planner
of Proposition~\ref{prop:two_impulse} with the matched-control
form~\eqref{eq:departure_impulse_matched} applied naively.

\begin{table}[t]
\caption{Forced and naive planners under mismatched control. $\|b_p\|$ is the
         position part of the forcing integral~\eqref{eq:forcing_integral};
         $\mathcal{C}$ is the total impulse of the forced and naive plans; the
         separation is the maximum along the forced transfer.}
\label{tab:forced}
\centering
\small
\setlength{\tabcolsep}{4pt}
\begin{tabular}{llrrrrr}
\hline\hline
& & \multicolumn{2}{c}{Terminal error, m} & & & \\
\cline{3-4}
Case & Mismatch & Forced & Naive & $\|b_p\|$, m & $\mathcal{C}$, m/s & Sep., m \\
\hline
S3, 90 min      & chief thrusts      & $1.91$  & $1434$ & $1434$ & $0.62$ / $0.15$  & $390$  \\
S1, 5 min       & chief thrusts      & $0.53$  & $445$  & $445$  & $3.16$ / $0.86$  & $145$  \\
S1, quarter orbit & chief thrusts    & $203.7$ & $9337$ & $9603$ & $16.23$ / $1.22$ & $3065$ \\
S1, quarter orbit & $60$~s slew      & $12.9$  & $551$  & $546$  & $2.23$ / $1.22$  & $437$  \\
S1, quarter orbit & both             & $189.0$ & $9337$ & $9607$ & $16.26$ / $1.22$ & $3075$ \\
S1, quarter orbit & $10\%$ misalignment & $7.50$ & $1044$ & $1058$ & $2.64$ / $1.22$ & $492$ \\
\hline\hline
\end{tabular}
\end{table}

The naive planner's error tracks $\|b_p\|$ case by case, and the forced planner
removes it: within the proximity-operations envelope it reaches $1.91$~m at
geosynchronous altitude and $0.53$~m on a five-minute low-Earth-orbit transfer,
against $1434$ and $445$~m. The quarter-orbit low-Earth-orbit cases are stress
cases, carrying the deputy $3.1$~km from the reference at a cost of $16$~m/s;
their residual is again $\delta_{\mathrm{att}}$ on the excursion-inflated
separation, and under sustained mismatch the transfer should be kept as short as
practical. Of the two ways the forced planner differs from the naive one, the
forcing $b$ does most of the work, removing $9040$ of the $9337$~m on the
chief-thrusting S1 case. The re-centered matrix of Remark~\ref{rem:mean_input},
which halves the coupling block there
($\|\tilde\Phi_{pR}\|/\|\Phi_{pR}\| = 0.5000$), removes a further $93$~m in
combination with the forcing but adds $128$~m without it.

The residuals of Theorem~\ref{thm:control_mismatch} are evaluated against
\eqref{eq:epsp_bound}--\eqref{eq:epsR_bound} at every timestep in
Table~\ref{tab:residuals}. All bounds hold.

\begin{table}[t]
\caption{Residual bounds of Theorem~\ref{thm:control_mismatch} under mismatched
         control; ratios are pointwise maxima over the transfer, and $t^\star$ is
         the instant of the $\varepsilon_v$ maximum. Dashes denote blocks whose
         bound in~\eqref{eq:epsp_bound}--\eqref{eq:epsR_bound} is identically
         zero.}
\label{tab:residuals}
\centering
\begin{tabular}{clccccc}
\hline\hline
& & \multicolumn{3}{c}{Ratio to bound} & & \\
\cline{3-5}
Case & Mismatch & $\varepsilon_p$ & $\varepsilon_v$ & $\varepsilon_R$
     & $t^\star/T$ & $\|\xi_R\|$ range \\
\hline
S1 & $\Delta a$              & ---      & $0.9999997$  & ---      & $0.499$ & $5.00^\circ$ \\
S1 & $\Delta\omega$          & $0.6041$ & $0.5935$     & $0.8165$ & $0.000$ & $5.00$--$8.93^\circ$ \\
S1 & both                    & $0.6868$ & $0.99999999$ & $0.8165$ & $0.244$ & $5.00$--$8.93^\circ$ \\
S3 & $\Delta a$              & ---      & $0.9502$     & ---      & $1.000$ & $5.00^\circ$ \\
S1 & $10\%$ misalignment     & ---      & $0.8190$     & ---      & $1.000$ & $5.00^\circ$ \\
\hline\hline
\end{tabular}
\end{table}

Two of the bounds are attained at interior instants of the transfer. Under matched
angular velocity $\xi_R$ cones about the orbit normal, so its angle to the
body-fixed acceleration mismatch passes through $90^\circ$ at mid-transfer, the
configuration in which Remark~\ref{rem:sharpness_control} predicts equality
in~\eqref{eq:epsv_bound}; the $\varepsilon_v$ ratio reaches $0.9999997$ there, and
$0.99999999$ with both mismatches present. The remaining ratios fall below one
because the corresponding angle never reaches $90^\circ$, and over $250$ attitude
axes on the chief-thrusting case the $\varepsilon_v$ ratio has a median of
$0.99999$ and never exceeds one. The sharpness of
Remark~\ref{rem:sharpness_control} is therefore realized along ordinary transfers,
not only in principle.

On the slew cases $\|\xi_R\|$ grows from $5.00^\circ$ to $8.93^\circ$, so the bound
coefficients~\eqref{eq:alpha_beta_def} must be evaluated along the trajectory
rather than fixed at $t_0$ as under matched control. A sustained rate mismatch of
$1.45\times10^{-3}$~rad/s drives the relative rotation to $\pi$ within $80$~min,
where the logarithmic coordinate saturates and the regime of
Remark~\ref{rem:validity_control} is left. This is slower than
$\pi/\|\Delta\omega\| = 36$~min because the body-fixed mismatch axis is carried
about the orbit normal through the $-\skewmat{\bar\omega}\xi_R$ term
of~\eqref{eq:xi_R_dot}, so the relative rotation precesses rather than
accumulating about a fixed axis.

\subsection{Gravity Linearization Bounds}\label{sec:num_gravity}

Table~\ref{tab:gravity} evaluates the attitude-induced error
$\|\delta_{\mathrm{att}}(t)\|$ and the higher-order gravity error
$\|\delta_{\mathrm{grav}}(t)\|$ at every timestep against the analytic bounds
of~\eqref{eq:bound_att}--\eqref{eq:bound_grav}, and reports each as a fraction of
the retained tidal term $\|G(t)\xi_p\|$. The remainder $\delta_{\mathrm{grav}}$ is
evaluated in a cancellation-free form from the nonlinear states.

\begin{table}[t]
\caption{Bound validation for Theorem~\ref{thm:gravity_linearization}; all cases
         at $\|\xi_R(0)\| = 5^\circ$. Every entry is a maximum over the
         trajectory, except the last column, which is the smallest angle between
         the relative position and the radial direction.}
\label{tab:gravity}
\centering
\small
\setlength{\tabcolsep}{4pt}
\begin{tabular}{lrrcccc}
\hline\hline
& \multicolumn{2}{c}{Fraction of $\|G\xi_p\|$}
& \multicolumn{2}{c}{Ratio to bound} & & \\
\cline{2-3}\cline{4-5}
Case & $\delta_{\mathrm{att}}$ & $\delta_{\mathrm{grav}}$
     & $\delta_{\mathrm{att}}$ & $\delta_{\mathrm{grav}}$ & $d/r$ & Closest to radial \\
\hline
S1, circular chemical & $5.72\times10^{-2}$ & $1.55\times10^{-4}$ & $0.9799$ & $0.9955$ & $1.2\times10^{-4}$ & $3.7^\circ$ \\
S2, circular electric & $6.56\times10^{-2}$ & $8.25\times10^{-5}$ & $0.9787$ & $0.9856$ & $7.8\times10^{-5}$ & $6.8^\circ$ \\
S3, geosynchronous    & $5.27\times10^{-2}$ & $2.01\times10^{-5}$ & $0.9824$ & $0.9950$ & $1.6\times10^{-5}$ & $4.0^\circ$ \\
S5, elliptical        & $8.13\times10^{-2}$ & $1.89\times10^{-4}$ & $0.9782$ & $0.9858$ & $1.7\times10^{-4}$ & $6.8^\circ$ \\
S1, chief thrusts     & $7.56\times10^{-2}$ & $6.69\times10^{-4}$ & $0.9769$ & $0.7719$ & $4.5\times10^{-4}$ & $29.2^\circ$ \\
S3, chief thrusts     & $7.50\times10^{-2}$ & $1.32\times10^{-5}$ & $0.9538$ & $0.8537$ & $9.2\times10^{-6}$ & $22.7^\circ$ \\
S1, $\|\xi_p(0)\| = 100$~km  & $6.17\times10^{-2}$ & $5.08\times10^{-2}$ & $0.9763$ & $0.9399$ & $0.051$ & $7.6^\circ$ \\
S1, $\|\xi_p(0)\| = 1000$~km & $6.23\times10^{-2}$ & $0.344$             & $0.9755$ & $0.6380$ & $0.439$ & $7.9^\circ$ \\
\hline\hline
\end{tabular}
\end{table}

All bounds are satisfied in every case. The $\delta_{\mathrm{att}}$ bound is tight
on every row, with maximum ratios between $0.954$ and $0.982$. The
$\delta_{\mathrm{grav}}$ bound is tight when the relative position passes close
to radial and looser otherwise: its ratio reaches $0.9955$ on S1 and $0.9950$ on
S3, where the relative position comes within $4^\circ$ of radial, and lies
between $0.64$ and $0.99$ on the remaining rows. The S1 maximum is the radial
worst case of Lemma~\ref{lem:gravity_remainder} realized: it occurs a quarter of
the way through the transfer, where the relative position is $3.7^\circ$ from
radial and the directional factor $5c_v^4 - 2c_v^2 + 1$ of~\eqref{eq:hessnormsq}
reaches $3.966$ of its maximum $4$. The same bounds hold over the $200$ attitude
axes of \ref{app:axis} and over the $1000$ ensemble trajectories of
Section~\ref{sec:num_setup}.

The attitude-induced fraction is nearly constant across scenarios, between
$5.27\times10^{-2}$ and $8.13\times10^{-2}$. This is structural: by
\eqref{eq:delta_att} the ratio $\|\Psi\xi_p\|/\|G(t)\xi_p\|$ equals
$c\,\theta + O(\theta^3)$, with $c$ depending only on the attitude axis and on
the direction of $\xi_p$ relative to the tidal axes, independent of orbital
radius, separation, and thrust, and bounded by $c \le 3/2$; at the nominal
geometry $c = 0.603$. This gives the design rule
\begin{equation}
    \|\delta_{\mathrm{att}}\| \;\le\; \tfrac{3}{2}\,\theta\,\|G(t)\,\xi_p\| ,
    \label{eq:delta_att_rule}
\end{equation}
so the attitude-induced gravity error is at most $13\%$ of the tidal term at
$5^\circ$ and $2.6\%$ at $1^\circ$, whatever the orbit, separation, or thrust. The
fractions observed in Table~\ref{tab:gravity}, $5\%$ to $8\%$ at $5^\circ$, lie
well inside it.

Figure~\ref{fig:gravity_bounds} plots the time histories of both remainders
against their respective bounds. The two remainders exchange dominance with
separation. At proximity-operations scales $\|\delta_{\mathrm{grav}}\|$ is two
to four orders of magnitude below $\|\delta_{\mathrm{att}}\|$. On S1 at $5^\circ$
the two are equal at an initial separation of $124$~km for a deputy drifting from
$\xi(t_0)$ without impulses, and at $297$~km on a planned transfer, whose
departure burn suppresses the excursion; on S3 the corresponding separations are
$1139$ and $1836$~km. The crossover scales approximately linearly in $\theta$.
Below it the attitude error rather than the separation is the dominant
linearization error, which is the entire regime of interest. At
$\|\xi_p(0)\| = 1000$~km the bound of Lemma~\ref{lem:gravity_remainder} is
tightest early, at a ratio of $0.638$ when $d/r = 0.16$. By arrival, at
$d/r = 0.44$, it has grown to $2.2$ times the tidal term against an actual
remainder of $0.34$ times it, and is satisfied at only $16\%$. The bound holds
everywhere, but it is the bound rather than the linearization that ceases to be
informative as $d$ approaches $r$.

\begin{figure}[t]
    \centering
    \includegraphics[width=\textwidth]{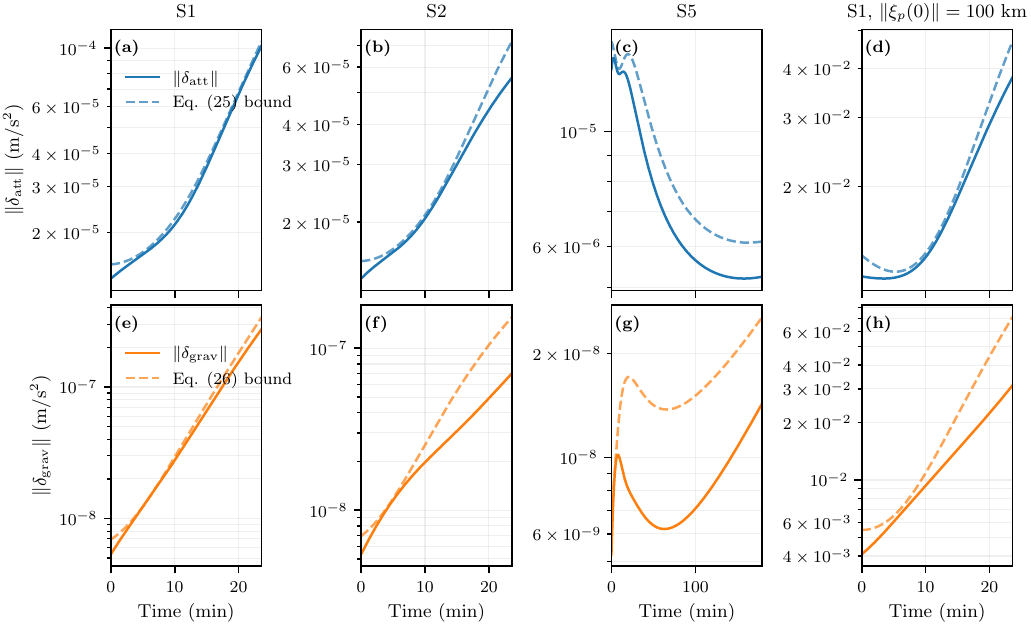}
    \caption{True linearization errors versus the analytic bounds of
             Theorem~\ref{thm:gravity_linearization}: $\delta_{\mathrm{att}}$
             (top) and $\delta_{\mathrm{grav}}$ (bottom) for (a,e) S1, (b,f) S2,
             (c,g) S5, and (d,h) S1 at $\|\xi_p(0)\| = 100$~km.}
    \label{fig:gravity_bounds}
\end{figure}

\subsection{Applicability Envelope and Replanning}\label{sec:num_envelope}

The planner requires the reference inputs and the relative state. Three sources
of error in those quantities are considered, in decreasing order of how binding
they are.

Knowledge of the reference thrust magnitude is the tightest requirement. An error
in the assumed $\|\bar a\|$ propagates through the coupling block
$-\skewmat{\bar a}$ of~\eqref{eq:A_matrix}; perturbing only the planner's
ephemeris by the same error, a $617$~m offset at $5\%$, has no effect at the
$10^{-4}$ level. At $5^\circ$ on S1 each percent of error shifts the terminal miss
vector by $4.9$~m, so the terminal error becomes $9.6$ to $10.9$~m at $\pm1\%$
and $24.6$ to $27.3$~m at $\pm5\%$, while at $\theta = 0$ it stays between $1.6$
and $4.3$~mm. Keeping the linearization residual dominant therefore requires the
reference thrust to be known to about $1.9\%$.

Knowledge of the relative attitude is a weaker requirement in practice. Terminal
error grows linearly, at a median over error directions of $115$~m per degree on
S1 and $18$~m per degree on S3, so about a quarter of the coupling correction is
forfeited per degree of knowledge error and $1^\circ$ retains roughly
three quarters of it. Star-tracker attitude determination delivers
arcsecond-class knowledge, so in cooperative operations where both vehicles
share attitude over a crosslink this condition is met with substantial margin.
Knowledge errors in relative position and velocity are benign: they map exactly
linearly, through $\tilde\Phi_{pp}$ and $\tilde\Phi_{pv}$ respectively, at a
typical gain of $1.44$~m per metre and $1.24$~m per millimetre per second in the
tested directions, with worst cases bounded by $\|\tilde\Phi_{pp}\| = 3.65$ and $\|\tilde\Phi_{pv}\| = 2580$~s.

Finite burn duration is accommodated by centering the burn arc on the impulse
epoch, as anticipated in Remark~\ref{rem:impulsive}. With the thrust direction
held in the deputy body frame, a centered departure burn changes the terminal
error by less than $10\%$ of the impulsive floor up to $7.8\%$ of the transfer
horizon, and a centered arrival burn up to $10\%$ of it. The change can have
either sign; in the flown geometry it opposes the floor's miss vector, so the
departure burn lowers the terminal error. The binding mechanism is the one that
remark identifies, the rotation of the reference frame during the burn, which
exceeds the finite-burn effect itself by a factor of $2.4$ for burns longer
than $5$~s and largely disappears when the thrust direction is held inertially rather than
in the body frame. Uncentered burns are an order of magnitude worse.

Applied simultaneously at realistic magnitudes, with per-epoch attitude
knowledge errors up to $0.3^\circ$, thrust-magnitude errors up to $2\%$, state
knowledge errors of $1$~m and $1$~mm/s, and departure burns up to $60$~s, these
perturbations give a median single-shot terminal error of $19.4$~m and a
ninetieth percentile of $35.6$~m, dominated by the attitude knowledge error at the
departure epoch. Because $\tilde\Phi$ and $b$ are cached, replanning from an
updated navigation solution is nearly free, and two interior replans reduce the
median to $1.69$~m and the ninetieth percentile to $2.87$~m while adding a median
of $6\%$ to the total impulse magnitude. Figure~\ref{fig:envelope} shows the two
distributions. Each replan removes the error injected at earlier epochs, so only
the last epoch's error survives, attenuated by the remaining horizon to $7.8\%$
of its single-shot sensitivity on S1 and $6.5\%$ on S3. Under perfect knowledge
the same two replans reduce the terminal error from $9.02$~m to $0.59$~m for a
$3.2\%$ increase in total impulse on S1, and from $0.523$~m to $9.1$~mm for
$0.2\%$ on S3; a third replan gains nothing further. The translational-only
planner, replanned at the same epochs, still ends at $37.8$~m on S1 and
$5.40$~m on S3, at $32\%$ and $15\%$ more impulse, because the coupling error
re-enters over every segment. The accuracy floor of Table~\ref{tab:mechanism} is
therefore a property of single-shot planning rather than of the framework.

\begin{figure}[t]
    \centering
    \includegraphics[width=0.6\textwidth]{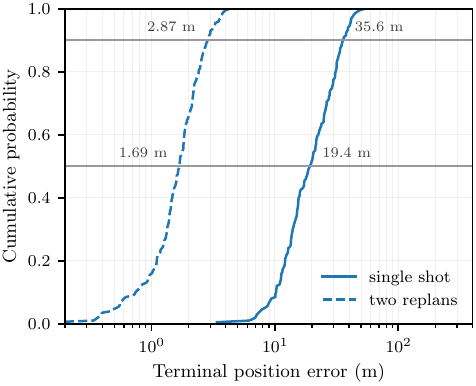}
    \caption{Terminal position error on Row~C under simultaneous knowledge and
             burn-duration perturbation ($200$ runs), single shot and with two
             replans.}
    \label{fig:envelope}
\end{figure}

Three perturbations are not represented in~\eqref{eq:loglinear}, lying outside
both $\tilde\nu$ and $\tilde m$ of~\eqref{eq:tildes}. The largest in low Earth
orbit is the second zonal harmonic, which contributes $0.64$~m to Row~C on an
equatorial orbit, $7\%$ of the linearization residual, and $0.35$~m at
$51.6^\circ$ inclination; at zero attitude error it raises the floor of Row~B from
$2.2$~mm to $0.49$~m. Its effect is the differential gradient missing from
$G(t)$, not an ephemeris error, and adding the $J_2$ Hessian to the tidal tensor,
with a $J_2$ reference ephemeris, removes $99.998\%$ of it at zero attitude error,
restoring the $2.2$~mm floor, and $94\%$ at $5^\circ$. The system remains linear
time-varying, so this extension costs nothing structurally. Differential drag
contributes $4.5$~cm at $500$~km, about a fourteenth of the $J_2$ effect, but
$1.19$~m at $300$~km and about ten times that at solar maximum, where it would
exceed the coupling residual. Differential solar radiation pressure contributes
$0.149$~m at geosynchronous altitude, $29\%$ of the residual there. All three grow
with the horizon, and over a day differential drag alone displaces the deputy by
$15.5$~km at $300$~km, which places formation-keeping timescales within the scope
of future reachable-set work.


\section{Conclusions}
\label{sec:conclusions}

This paper has developed a two-impulse rendezvous planning framework from
log-error dynamics on $\mathrm{SE}_2(3)$ with linearized gravity and linearized
control mismatch. The gravity mismatch is decomposed into an attitude-induced
component $\delta_{\mathrm{att}}$ and a higher-order gravity component
$\delta_{\mathrm{grav}}$, each with an analytic bound that the simulations show to
be tight in the configurations that attain it; both are small fractions of the
first-order tidal term at proximity-operations scales.

Under circular coasting Keplerian assumptions the resulting linear time-varying
system recovers the HCW equations exactly, without additional approximation,
which establishes HCW as a special case of the $\mathrm{SE}_2(3)$ framework rather
than an independent model. Relaxing the matched-control assumption, the
body-frame input mismatch is linearized to first order with sharp residual bounds
in each block, and absorbing the linear part re-centers the system matrix at the
mean of the two vehicles' inputs.

Under matched control the resulting planner carries the thrust-induced
attitude--translation coupling of a thrusting reference exactly, as off-diagonal blocks of a state transition matrix that does
not depend on the attitude error. That coupling is absent from HCW-based planners
and equally from the Yamanaka--Ankersen state transition matrix, which treats
arbitrary eccentricity exactly under coasting yet fails as the
translational-only planner does, to two parts in $10^5$, once attitude error meets
thrust. Its price is the
attitude-induced gravity residual $\delta_{\mathrm{att}}$, which is the planner's
accuracy floor. A Cartesian planner of the same structure pays instead a residual
quadratic in the attitude error, and the ratio of the two is of order $\kappa$.
The coupling ratio $\kappa$, computable before planning, therefore classifies the
regimes: where $\kappa \gtrsim 1$ the proposed planner reduces the error of
coupling-blind planners by more than an order of magnitude, and above a crossover
near $\kappa = 1.5$ it is also the more accurate attitude-independent planner;
below that crossover a Cartesian planner, and at $\kappa \ll 1$ in some
geometries a classical one, is preferable. Two replans from the cached
transition matrix reduce the single-shot floor by more than an order of
magnitude.

Future work will investigate robustness to modeling uncertainty using LMI-based
reachable-set analysis; the retention of the conjugated tidal tensor, or of
second-order state transition tensors, to remove $\delta_{\mathrm{att}}$; and the
recovery of the Yamanaka--Ankersen state transition matrix as a special case for
elliptical coasting references, which the simulations of \ref{app:axis}
support numerically.



\section*{Declaration of generative AI and AI-assisted technologies in the manuscript preparation process}
Generative AI (Codex, Claude) was used to aid development of the simulation code and help with editing the paper. However, the authors have verified all technical content, derivations, and results and take full responsibility.

\appendix

\section{Singular Values of the $\mathrm{SO}(3)$ Left Jacobian}
\label{app:jacobian_proof}

We work from the closed form $J_\ell^{\mathrm{SO}(3)} = I + \gamma_1 S + \gamma_2 S^2$
\cite[\S7.1.5]{barfoot2024}, a polynomial in the skew matrix $S = [\xi_R]_\times$.
Since $S$ is skew, $\mathbb{R}^3$ splits orthogonally into the axis
$\mathrm{span}\{\hat\xi\}$, $\hat\xi = \xi_R/\theta$, and the plane
$P = \hat\xi^\perp$; both are invariant under $S$ (hence under any polynomial
in $S$), so we analyze $J_\ell^{\mathrm{SO}(3)}$ on each separately.

\emph{Axis.} $S\hat\xi = \xi_R \times \hat\xi = 0$, so $S^2\hat\xi = 0$ and
$J_\ell^{\mathrm{SO}(3)}\hat\xi = \hat\xi$. The axis contributes the singular
value $1$.

\emph{Plane.} Choose an orthonormal basis $\{w, w'\}$ of $P$ with $w' =
\hat\xi \times w$ (so $\{\hat\xi, w, w'\}$ is right-handed). Then, using the
cross-product identity $\hat\xi \times (\hat\xi \times w) = \hat\xi(\hat\xi^\top
w) - w = -w$ for $w \perp \hat\xi$,
\begin{equation}
    Sw = \xi_R \times w = \theta\, w', \qquad
    Sw' = \xi_R \times (\hat\xi \times w) = \theta\,(\hat\xi \times w')
        = -\theta\, w.
    \label{eq:app_plane_action}
\end{equation}
Thus on $P$, in the basis $\{w, w'\}$, $S$ acts as $\theta
\left[\begin{smallmatrix}0 & -1\\ 1 & 0\end{smallmatrix}\right]$---that is, as
$\theta$ times a $90^\circ$ rotation---and consequently $S^2 = -\theta^2 I$ on
$P$. This last relation is exactly what allows the plane to be identified with
$\mathbb{C}$ via $w \leftrightarrow 1$, $w' \leftrightarrow i$, under which
$S \leftrightarrow i\theta$; $S^2 \leftrightarrow (i\theta)^2 = -\theta^2$ is
consistent with $S^2 = -\theta^2 I$. Under this identification
$J_\ell^{\mathrm{SO}(3)}$ acts on $P$ as multiplication by the complex number
\begin{equation}
    c_\mathbb{C} = 1 + \gamma_1(i\theta) + \gamma_2(i\theta)^2
      = (1 - \gamma_2\theta^2) + i\,\gamma_1\theta
      = \tfrac{\sin\theta}{\theta} + i\,\tfrac{1-\cos\theta}{\theta},
    \label{eq:app_c}
\end{equation}
where the last equality uses $\gamma_1\theta = (1-\cos\theta)/\theta$ and $1 -
\gamma_2\theta^2 = 1 - (\theta - \sin\theta)/\theta = \sin\theta/\theta$. Its modulus
is
\begin{equation}
    |c_\mathbb{C}| = \frac{1}{\theta}\sqrt{\sin^2\theta + (1-\cos\theta)^2}
        = \frac{1}{\theta}\sqrt{2 - 2\cos\theta}
        = \frac{2\sin(\theta/2)}{\theta},
    \label{eq:app_cmod}
\end{equation}
using $\sin^2\theta + (1-\cos\theta)^2 = 2 - 2\cos\theta$ and the half-angle
identity $2 - 2\cos\theta = 4\sin^2(\theta/2)$. Its argument follows from the
same two identities as
$\arg c_\mathbb{C} = \arctan\bigl[(1-\cos\theta)/\sin\theta\bigr] = \theta/2$,
exactly and not merely to leading order; this is the rotation invoked
in~\eqref{eq:impulse_norm_bounds} and in Section~\ref{sec:num_mechanism}.

\emph{Singular values on the plane.} Multiplication by
$c_\mathbb{C} = |c_\mathbb{C}|\,e^{i\arg c_\mathbb{C}}$
is the composition of a rotation (by $\arg c_\mathbb{C}$, an isometry) and a uniform
scaling by $|c_\mathbb{C}|$; equivalently, as a real $2\times 2$ map it is
$|c_\mathbb{C}|\left[\begin{smallmatrix}\cos\phi & -\sin\phi\\ \sin\phi &
\cos\phi\end{smallmatrix}\right]$ with $\phi = \arg c_\mathbb{C}$, whose singular values
are both $|c_\mathbb{C}|$. Hence $J_\ell^{\mathrm{SO}(3)}$ has both plane singular values
equal to $|c_\mathbb{C}|$.

Combining the axis and plane contributions, the singular values of
$J_\ell^{\mathrm{SO}(3)}$ are $\{1,\, |c_\mathbb{C}|,\, |c_\mathbb{C}|\} = \{1,\,
2\sin(\theta/2)/\theta,\, 2\sin(\theta/2)/\theta\}$, which
establishes~\eqref{eq:jacobian_svals}. \qed

\section{Commutator Norms for the Attitude-Induced Gravity Bound}
\label{app:datt}

This appendix derives the two commutator norms~\eqref{eq:commutator_norms} used
in the bound on $\delta_{\mathrm{att}}$. In the body frame $G(t) =
\bar{R}^\top G_I(t)\bar{R} = \tfrac{\mu}{r^3}(3\hat{r}\hat{r}^\top - I)$, which
is symmetric with eigenvalues $\tfrac{\mu}{r^3}\{2,-1,-1\}$, so $\|G(t)\| =
2\mu/r^3$. Its isotropic part commutes with $S$, so
\begin{equation}
    [G(t),B] = \tfrac{3\mu}{r^3}[\hat{r}\hat{r}^\top,B]
    \qquad\text{for any } B \in \mathbb{R}^{3\times3}.
    \label{eq:app_comm_identity}
\end{equation}

\emph{First commutator.} With $z := S \hat{r} = \xi_R \times \hat{r}$ (so
$z \perp \hat{r}$ and $\|z\| = \theta\sin\varphi$), using
$S\,\hat{r}\hat{r}^\top = (S\hat{r})\hat{r}^\top$ and its transpose
$\hat{r}\hat{r}^\top S = -\hat{r}(S\hat{r})^\top$ (from $S^\top = -S$),
\begin{equation}
    [G(t),S] = \tfrac{3\mu}{r^3}[\hat{r}\hat{r}^\top,S]
    = -\tfrac{3\mu}{r^3}\bigl(\hat{r} z^\top + z \hat{r}^\top\bigr).
    \label{eq:app_comm1}
\end{equation}
The symmetric matrix $\hat{r} z^\top + z \hat{r}^\top$ acts only on
$\mathrm{span}\{\hat{r},\hat z\}$ ($\hat z := z/\|z\|$), where in that
orthonormal basis it equals
$\|z\|\left[\begin{smallmatrix}0&1\\1&0\end{smallmatrix}\right]$, with
eigenvalues $\pm\|z\|$. Hence
$\|[G(t),S]\| = \tfrac{3\mu}{r^3}\|z\| = \tfrac{3\mu}{r^3}\theta\sin\varphi$.

\emph{Second commutator.} Since $S^2 = \theta^2(\hat\xi\hat\xi^\top - I)$ and
the central $-I$ commutes with $G(t)$, using $\hat{r}^\top\hat\xi = \cos\varphi$,
\begin{equation}
    [G(t),S^2] = \tfrac{3\mu}{r^3}\theta^2[\hat{r}\hat{r}^\top,\hat\xi\hat\xi^\top]
    = \tfrac{3\mu}{r^3}\theta^2\cos\varphi\,
      \bigl(\hat{r}\hat\xi^\top - \hat\xi \hat{r}^\top\bigr).
    \label{eq:app_comm2}
\end{equation}
Resolving $\hat\xi = \cos\varphi\,\hat{r} + \sin\varphi\,w$ with $w \perp \hat{r}$
unit, the symmetric $\cos\varphi\,\hat{r}\hat{r}^\top$ contribution cancels in
the antisymmetric bracket, leaving $\hat{r}\hat\xi^\top - \hat\xi \hat{r}^\top =
\sin\varphi\,(\hat{r}w^\top - w\hat{r}^\top)$. On $\mathrm{span}\{\hat{r},w\}$
this acts as
$\left[\begin{smallmatrix}0&1\\-1&0\end{smallmatrix}\right]$, of unit spectral
norm. Therefore, using $2\sin\varphi\cos\varphi = \sin2\varphi$,
$\|[G(t),S^2]\| = \tfrac{3\mu}{r^3}\theta^2|\cos\varphi|\sin\varphi =
\tfrac{3\mu}{2r^3}\theta^2|\sin2\varphi|$. This
establishes~\eqref{eq:commutator_norms}. \qed

\section{Monotonicity of the Inverse-Jacobian Coefficient}
\label{app:ma_mp}

This appendix establishes the inequality $m_a > 2m_p$ used in the proof of
Lemma~\ref{lem:M_bound}, where $m_a = \beta'$ and $m_p = \beta/\theta$
by~\eqref{eq:m_a} and~\eqref{eq:m_p}.

Differentiating $\beta(\theta) = 1 - (\theta/2)\cot(\theta/2)$ gives
$\beta' = -\tfrac12\cot\tfrac{\theta}{2} + \tfrac{\theta}{4}\csc^2\tfrac{\theta}{2}$, so
\begin{equation}
    m_a - 2m_p = \beta' - \frac{2\beta}{\theta}
    = \tfrac{1}{2}\cot\tfrac{\theta}{2}
      + \tfrac{\theta}{4}\csc^2\tfrac{\theta}{2} - \frac{2}{\theta}
    = \frac{h(\theta)}{4\,\theta\,\sin^2(\theta/2)},
    \label{eq:app_ma_mp_identity}
\end{equation}
where, using $2\sin\tfrac{\theta}{2}\cos\tfrac{\theta}{2} = \sin\theta$ and
$2\sin^2\tfrac{\theta}{2} = 1 - \cos\theta$,
\begin{equation}
    h(\theta) := \theta^2 + \theta\sin\theta + 4\cos\theta - 4.
    \label{eq:app_h_def}
\end{equation}
The denominator of~\eqref{eq:app_ma_mp_identity} is positive on $(0,\pi)$, so it suffices to
show $h > 0$ there. Successive differentiation gives
\begin{equation}
    h' = 2\theta + \theta\cos\theta - 3\sin\theta,
    \qquad
    h'' = 2 - 2\cos\theta - \theta\sin\theta,
    \qquad
    h''' = \sin\theta - \theta\cos\theta,
    \label{eq:app_h_derivs}
\end{equation}
with $h(0) = h'(0) = h''(0) = h'''(0) = 0$. On $[0,\pi/2]$ one has
$h''' = \cos\theta\,(\tan\theta - \theta) \ge 0$, since $\cos\theta \ge 0$ and
$\tan\theta \ge \theta$ there; on $(\pi/2,\pi]$ both $\sin\theta \ge 0$ and
$-\theta\cos\theta \ge 0$. Hence $h''' \ge 0$ on $[0,\pi]$, with strict inequality on
$(0,\pi)$. Integrating three times from zero gives $h'' \ge 0$, then $h' \ge 0$, then
$h > 0$ on $(0,\pi)$, and therefore $m_a > 2m_p$ there.

Equivalently, since $m_a - 2m_p = \theta^2 c'(\theta)$ by differentiation of
$c = \beta/\theta^2$, the inequality states that the coefficient $c$
of~\eqref{eq:Jl_inv_so3} is strictly increasing on $(0,\pi)$. As $\theta \to 0^+$,
$h(\theta) = \theta^6/360 + O(\theta^8)$, so the ratio $m_a/m_p$ approaches $2$ from above and
the inequality is tight in the small-angle limit. \qed


\section{Attitude Axis, Thrust Direction, and Eccentricity}
\label{app:axis}

The nominal state~\eqref{eq:nominal_state} fixes the attitude-error axis, the
thrust direction, and, through the scenario, the eccentricity. Each affects the
terminal error at fixed attitude magnitude, and this appendix reports the
distributions.

Holding $\|\xi_R(0)\| = 5^\circ$ and sweeping the axis over a $200$-point
Fibonacci lattice on the sphere, the S1 terminal error of the proposed planner
ranges from $1.83$ to $39.9$~m with a median of $23.26$~m, a $22$-fold spread
(Fig.~\ref{fig:axis}). The nominal axis $[1,1,1]^\top/\sqrt3$ gives $9.02$~m, at
the thirteenth percentile, so it is a favorable draw rather than a typical one.
At the median-ranked axis the proposed planner reaches $23.2$~m against $786.2$~m for
the coupling-blind planners, a $97.0\%$ correction, and this row is reported in
Table~\ref{tab:mechanism}; on S3 the median axis gives $0.626$~m against
$113.5$~m, so the nominal axis is favorable by a factor of $2.6$ on S1 but only
$1.2$ at geosynchronous altitude. The bound ratio of~\eqref{eq:bound_att} remains
between $0.971$ and $0.99996$ across all $200$ axes, so the direction dependence
retained in~\eqref{eq:bound_att} through $\varphi$ tracks the true
attitude-induced error for every axis and not only for the tabulated one. The
ratio is positively correlated with the terminal error (Pearson coefficient
$0.84$), as expected if $\delta_{\mathrm{att}}$ sets the residual. The axis
pattern itself is the intrinsic $\varphi$-dependence of $\delta_{\mathrm{att}}$:
on S3, where there is no excursion, the terminal error correlates with the angle
between the axis and the radial direction at $0.91$. The excursion mechanism of
Section~\ref{sec:num_attitude} sets the magnitude and the $\theta$-scaling
instead.

\begin{figure}[t]
    \centering
    \includegraphics[width=0.5\textwidth]{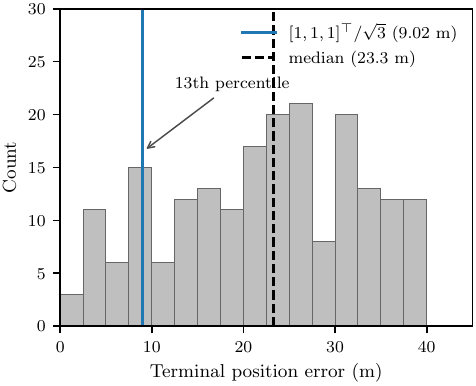}
    \caption{Distribution of the S1 terminal error of the proposed planner over
             $200$ attitude-error axes at $\|\xi_R(0)\| = 5^\circ$.}
    \label{fig:axis}
\end{figure}

Sweeping the thrust direction over a $100$-point body-frame lattice at fixed
attitude axis gives terminal errors from $1.52$ to $22.3$~m, a $15$-fold spread,
with a median of $11.0$~m; the translational-only planner ranges from $81$ to
$875$~m. The proposed plan costs between $0.38$ and $1.59$~m/s depending on
direction, while the translational-only plan costs $0.313$~m/s for every
direction, since it ignores the coupling. The nominal radial choice gives
$9.02$~m, at the thirty-ninth percentile. No single geometric factor explains
the pattern: the terminal error correlates with the squared cosines of the
angles between the thrust and the radial, along-track, normal, and attitude-axis
directions at no more than $0.31$ in magnitude.

Table~\ref{tab:ecc} varies the eccentricity at fixed perigee radius
($6906$~km), with a quarter-period horizon and $\theta = 5^\circ$. On coasting
references the proposed and Yamanaka--Ankersen planners agree to $10^{-8}$ at
every eccentricity, as two first-order-exact linearizations of the same Keplerian
relative motion should, while the classical planner degrades with eccentricity.
On thrusting references the Yamanaka--Ankersen error exceeds the proposed
planner's by a factor that grows from $8.6$ on the circular orbit to about $48$
at $e = 0.5$ and then levels off. The thrust--attitude channel and the
eccentricity channel are therefore separable.

\begin{table}[t]
\caption{Terminal position error versus eccentricity at fixed perigee radius,
         $\theta = 5^\circ$, m.}
\label{tab:ecc}
\centering
\begin{tabular}{crrrrrr}
\hline\hline
& & \multicolumn{3}{c}{Thrusting reference} & \multicolumn{2}{c}{Coasting reference} \\
\cline{3-5}\cline{6-7}
$e$ & Horizon, min & Proposed & YA & Ratio & Proposed $=$ YA & HCW \\
\hline
$0$    & $23.8$  & $5.15$ & $44.5$ & $8.6$  & $5.68$ & $5.68$ \\
$0.1$  & $27.9$  & $6.10$ & $68.3$ & $11.2$ & $7.20$ & $54.2$ \\
$0.3$  & $40.6$  & $8.37$ & $185$  & $22.2$ & $12.6$ & $257$  \\
$0.5$  & $67.3$  & $12.6$ & $604$  & $47.9$ & $26.1$ & $708$  \\
$0.74$ & $179.5$ & $93.9$ & $4456$ & $47.5$ & $109$  & $2736$ \\
\hline\hline
\end{tabular}
\end{table}

\bibliographystyle{elsarticle-num}
\bibliography{references}

\end{document}